\documentclass[
paper=letter,%
numbers=noendperiod,%
captions=nooneline,%
abstracton,%
DIV=10%
]{scrartcl}

\usepackage[T1]{fontenc}%
\usepackage{lmodern}%
\usepackage[american]{babel}%
\usepackage{microtype}%

\usepackage[
hyperref,%
table%
]{xcolor}%

\usepackage{scrlayer-scrpage}%
\usepackage{eso-pic}%
\usepackage{rotating}%
\usepackage{amsmath}%
\usepackage{mathtools}%
\usepackage{amssymb}%
\usepackage{amsthm}%
\usepackage{thmtools}%
\usepackage{etoolbox}%
\usepackage{bm}%
\usepackage{bbm}%
\usepackage{enumitem}%
\usepackage{graphicx}%
\usepackage{grffile}%
\usepackage{tikz}%
\usepackage{wrapfig}%
\usepackage{tabularx}%
\usepackage{siunitx}%
\usepackage{booktabs}%
\usepackage{multirow}%
\usepackage{fancyvrb}%
\usepackage{listings}%
\usepackage{csquotes}%
\usepackage[
 backend=bibtex,%
style=authoryear,%
dashed=false,%
hyperref=true,%
useprefix=true,%
maxcitenames=2,%
maxbibnames=6,%
seconds=true,%
date=iso,%
urldate=iso%
]{biblatex}%
\usepackage[
hypertexnames=false,%
setpagesize=false,%
pdfborder={0 0 0},%
pdfstartview=Fit,%
bookmarksopen=true,%
bookmarksnumbered=true%
]{hyperref}%
\xdefinecolor{lightgray}{RGB}{247, 247, 247}%
\xdefinecolor{semilightgray}{RGB}{240, 240, 240}%
\xdefinecolor{middlegray}{RGB}{127, 127, 127}%
\xdefinecolor{blue}{RGB}{58, 95, 205}%
\xdefinecolor{deepskyblue}{RGB}{0, 154, 205}%
\xdefinecolor{chocolate}{RGB}{205, 102, 29}%

\setkomafont{pageheadfoot}{\normalfont\normalcolor\sffamily}%
\setkomafont{pagenumber}{\normalfont\normalcolor\sffamily}%
\automark{section}%
\setkomafont{captionlabel}{\normalfont\normalcolor\sffamily\bfseries}%

\setlist{%
  align=left,%
  labelsep=*,%
  leftmargin=*,%
  topsep=1mm,%
  itemsep=0mm%
}
\newcommand*{\mysquare}{\rule[0.18em]{0.36em}{0.36em}}
\newcommand*{\mytriangle}{\raisebox{0.12em}{\resizebox{0.48em}{0.48em}{$\blacktriangleright$}}}
\newcommand*{\mybar}{\rule[0.32em]{0.62em}{0.08em}}
\newcommand*{\mydot}{\raisebox{0.14em}{\resizebox{0.44em}{!}{$\bullet$}}}
\setlist[itemize,1]{label={\mysquare\ }}%
\setlist[itemize,2]{label={\mytriangle\ }}%
\setlist[itemize,3]{label={\mybar\ }}%
\setlist[itemize,4]{label={\mydot\ }}%
\setlist[enumerate,1]{label=\arabic*)}%
\setlist[enumerate,2]{label=\arabic{enumi}.\arabic*)}%
\setlist[enumerate,3]{label=\arabic{enumi}.\arabic{enumii}.\arabic*)}%

\makeatletter
\newcommand\myisodate{\number\year-\ifcase\month\or 01\or 02\or 03\or 04\or 05\or 06\or 07\or 08\or 09\or 10\or 11\or 12\fi-\ifcase\day\or 01\or 02\or 03\or 04\or 05\or 06\or 07\or 08\or 09\or 10\or 11\or 12\or 13\or 14\or 15\or 16\or 17\or 18\or 19\or 20\or 21\or 22\or 23\or 24\or 25\or 26\or 27\or 28\or 29\or 30\or 31\fi}%
\makeatother
\newcommand*{\abstractnoindent}{}%
\let\abstractnoindent\abstract
\renewcommand*{\abstract}{\let\quotation\quote\let\endquotation\endquote
  \abstractnoindent}
\deffootnote[1em]{1em}{1em}{\textsuperscript{\thefootnotemark}}%
\lstdefinestyle{input}{
  backgroundcolor=\color{semilightgray},%
  commentstyle=\itshape\color{black},%
  keywordstyle=\color{black},%
  emphstyle=\color{black},%
  stringstyle=\color{black},%
  numbers=left,%
  numbersep=4.8pt,%
  numberstyle=\color{darkgray!80}\tiny%
}
\lstdefinestyle{output}{
  backgroundcolor=\color{lightgray}%
}

\lstdefinestyle{Rstyle}{
  language=R,%
  keywords={function, if, else, switch, repeat, while, for, in, next, break},%
  otherkeywords={},%
  emph={TRUE, FALSE, NULL, NA, NaN, Inf}%
}
\expandafter\let\csname Sinput\endcsname\relax
\expandafter\let\csname endSinput\endcsname\relax
\expandafter\let\csname Soutput\endcsname\relax
\expandafter\let\csname endSoutput\endcsname\relax
\lstnewenvironment{Sinput}[1][]{%
  \lstset{style=input, style=Rstyle}
  #1%
}{\vspace{-0.25\baselineskip}}%
\lstnewenvironment{Soutput}[1][]{%
  \lstset{style=output, style=Rstyle}
  #1%
}{\vspace{-0.25\baselineskip}}%

\lstdefinestyle{LaTeXstyle}{
  language=[LaTeX]TeX,%
  texcs={},%
  otherkeywords={}%
}
\lstnewenvironment{LaTeXinput}[1][]{%
  \lstset{style=input, style=LaTeXstyle}
  #1%
}{\vspace{-0.25\baselineskip}}%
\lstnewenvironment{LaTeXoutput}[1][]{%
  \lstset{style=output, style=LaTeXstyle}
  #1%
}{\vspace{-0.25\baselineskip}}%

\lstdefinestyle{otherstyle}{
  language={},%
  otherkeywords={},%
  upquote=true%
}
\lstnewenvironment{otherinput}[1][]{%
  \lstset{style=input, style=otherstyle}
  #1%
}{\vspace{-0.25\baselineskip}}%
\lstnewenvironment{otheroutput}[1][]{%
  \lstset{style=output, style=otherstyle}
  #1%
}{\vspace{-0.25\baselineskip}}%
\DeclareNameAlias{sortname}{family-given}%
\DefineBibliographyExtras{american}{\DeclareQuotePunctuation{}}%
\renewbibmacro*{volume+number+eid}{%
  \setunit*{\addcomma\space}%
  \printfield{volume}%
  \printfield{number}}
\DeclareFieldFormat*{number}{(#1)}
\DeclareFieldFormat*{title}{#1}%
\DeclareFieldFormat{doi}{%
  \ifhyperref
    {\href{http://dx.doi.org/#1}{\nolinkurl{doi:#1}}}%
    {\nolinkurl{doi:#1}}}%
\renewbibmacro*{in:}{}%
\DeclareFieldFormat{isbn}{ISBN #1}%
\DeclareFieldFormat{pages}{#1}%
\DeclareFieldFormat{url}{\url{#1}}%
\DeclareFieldFormat{urldate}{\mkbibparens{#1}}%
\renewcommand*{\cite}[2][]{\textcite[#1]{#2}}%

\newif\ifstarttheorem
\declaretheoremstyle[%
  spaceabove=0.5em,
  spacebelow=0.5em,
  headfont=\sffamily\bfseries\global\starttheoremtrue,
  notefont=\sffamily\bfseries,
  notebraces={(}{)},
  headpunct={},
  bodyfont=\normalfont,
  postheadspace=\newline%
]{myMainStyle}
\declaretheorem[style=myMainStyle, numberwithin=section]{definition}%
\declaretheorem[style=myMainStyle, sibling=definition]{proposition}

\declaretheorem[style=myMainStyle, sibling=definition]{remark}

\declaretheorem[style=myMainStyle, sibling=definition]{algorithm}

\makeatletter
\preto\itemize{%
  \if@inlabel
    \ifstarttheorem
      \mbox{}\par\nobreak\vskip\glueexpr-\parskip-\baselineskip+0.25em\relax\hrule\@height\z@
    \fi%
  \fi%
  \global\starttheoremfalse%
 \def\tempa{proof}%
 \ifx\tempa\mycurrenvir
    \ifstarttheorem
      \mbox{}\par\nobreak\vskip\glueexpr-\parskip-\baselineskip+0.25em\relax\hrule\@height\z@
    \fi%
 \fi%
 \global\starttheoremfalse%
}
\preto\enditemize{\global\starttheoremfalse}
\makeatother

\makeatletter
\preto\enumerate{%
  \if@inlabel
    \ifstarttheorem
      \mbox{}\par\nobreak\vskip\glueexpr-\parskip-\baselineskip+0.25em\relax\hrule\@height\z@
    \fi%
  \fi%
  \global\starttheoremfalse%
 \def\tempa{proof}%
 \ifx\tempa\mycurrenvir
    \ifstarttheorem
      \mbox{}\par\nobreak\vskip\glueexpr-\parskip-\baselineskip+0.25em\relax\hrule\@height\z@
    \fi%
 \fi%
 \global\starttheoremfalse%
}
\preto\endenumerate{\global\starttheoremfalse}
\makeatother

\newcommand*{\omu}[3]{\underset{#3}{\overset{#1}{#2}}}
\newcommand*{\T}{^{\top}}
\renewcommand*{\i}{\leftarrow}
\newcommand*{\isim}{\omu{\text{\tiny{ind.}}}{\sim}{}}

\newcommand*{\darrow}{\omu{\text{\tiny{d}}}{\rightarrow}{}}

\newcommand*{\IN}{\mathbb{N}}

\newcommand*{\IR}{\mathbb{R}}

\newcommand*{\IG}{\operatorname{IG}}
\newcommand*{\U}{\operatorname{U}}

\newcommand*{\N}{\operatorname{N}}

\newcommand*{\I}{\mathbbm{1}}
\newcommand*{\rd}{\mathrm{d}}

\newcommand*{\argmax}{\operatorname*{argmax}}
\newcommand*{\argmin}{\operatorname*{argmin}}

\renewcommand*{\P}{\mathbb{P}}
\newcommand*{\E}{\mathbb{E}}

\newcommand*{\var}{\operatorname{Var}}
\newcommand*{\Var}{\operatorname{Var}}

\newcommand*{\R}{\textsf{R}}
\newcommand*{\eps}{\varepsilon}
\newcommand*{\btheta}{\bm{\theta}}
\newcommand*{\ba}{\bm{a}}

\newcommand*{\bx}{\bm{x}}

\newcommand*{\bu}{\bm{u}}
\newcommand*{\bU}{\bm{U}}
\newcommand*{\bv}{\bm{v}}

\newcommand*{\bbeta}{\bm{\beta}}

\newcommand*{\bZ}{\bm{Z}}

\newcommand*{\bX}{\bm{X}}

\newcommand*{\bzero}{\bm{0}}

\newcommand*{\bmu}{\bm{\mu}}

\newcommand*{\gTopt}{g_T^{\text{\tiny{opt}}}}
\newcommand*{\gTopts}{g_T^{\text{\tiny{opt,s}}}}
\newcommand*{\hgTopts}{\hat{g}_T^{\text{\tiny{opt,s}}}}
\newcommand*{\hgTopt}{\hat{g}_T^{\text{\tiny{opt}}}}
\newcommand*{\hgToptu}{\hat{g}_{T,u}^{\text{\tiny{opt}}}}
\newcommand*{\muSIS}{\mu_{\text{\tiny{SIS}}}}
\newcommand*{\sigSSISsq}{\sigma_{\text{\tiny{SSIS}}}^2}
\newcommand*{\hsigSSISsq}{\hat{\sigma}_{\text{\tiny{SSIS}}}^2}
\newcommand*{\sigSISsq}{\sigma_{\text{\tiny{SIS}}}^2}

\newcommand*{\hmuMCn}{\hat{\mu}^\text{\tiny{MC}}_n}
\newcommand*{\hmuIS}{\hat{\mu}^{\text{\tiny{IS}}}}

\newcommand*{\hmuSISn}{\hat{\mu}^{\text{\tiny{SIS}}}_n}
\newcommand*{\hmuSSISn}{\hat{\mu}^{\text{\tiny{SSIS}}}_n}
\newcommand*{\hmuSISoptn}{\hat{\mu}^\text{\tiny{SIS,opt}}_n}
\newcommand*{\hmuSSISoptn}{\hat{\mu}^\text{\tiny{SSIS,opt}}_n}

\newcommand{\SIS}{\operatorname{SIS}}
\newcommand{\SSIS}{\operatorname{SSIS}}

\newcommand{\npil}{n_{\text{\tiny{pilot}}}}
\newcommand{\nknot}{n_{\text{\tiny{knot}}}}
\newcommand{\ntot}{n_{\text{\tiny{tot}}}}
\begin{document}
\thispagestyle{plain}
\begin{center}
  \sffamily
  {\bfseries\LARGE Single-Index Importance Sampling with Stratification
\par}
  \bigskip\smallskip
  {\Large
    Erik Hintz\footnote{Department of Statistics and Actuarial Science, University of
    Waterloo, 200 University Avenue West, Waterloo, ON, N2L
    3G1,
    \href{mailto:erik.hintz@uwaterloo.ca}{\nolinkurl{erik.hintz@uwaterloo.ca}}.},
    Marius Hofert\footnote{Department of Statistics and Actuarial Science, University of
    Waterloo, 200 University Avenue West, Waterloo, ON, N2L
    3G1,
    \href{mailto:marius.hofert@uwaterloo.ca}{\nolinkurl{marius.hofert@uwaterloo.ca}}. The
    author would like to thank NSERC for financial support for this work through Discovery
    Grant RGPIN-5010-2015.},
    Christiane Lemieux\footnote{Department of Statistics and Actuarial Science, University of
    Waterloo, 200 University Avenue West, Waterloo, ON, N2L
    3G1,
    \href{mailto:clemieux@uwaterloo.ca}{\nolinkurl{clemieux@uwaterloo.ca}}. The
    author would like to thank NSERC for financial support for this work through Discovery
    Grant RGPIN-238959.},
    Yoshihiro Taniguchi\footnote{Canadian Imperial Bank of Commerce,
   \href{mailto:ytanigucmc@gmail.com}{\nolinkurl{ytanigucmc@gmail.com}}}
    \par
    \bigskip
    \myisodate\par}
\end{center}
\par\smallskip

\begin{abstract}
In many stochastic problems, the output of interest depends on an input random vector mainly through a single random variable (or index) via an appropriate univariate transformation of the input. We exploit this feature by proposing an importance sampling method that makes rare events more likely by changing the distribution of the chosen index. Further variance reduction is guaranteed by combining this single-index importance sampling approach with stratified sampling. The dimension-reduction effect of single-index importance sampling also enhances the effectiveness of quasi-Monte Carlo methods. The proposed method applies to a wide range of financial or risk management problems. We demonstrate its efficiency for estimating large loss probabilities of a credit portfolio under a normal and $t$-copula model and show that our method outperforms the current standard for these problems. %
\end{abstract}
\minisec{Keywords}{Single-index model, importance sampling, stratified sampling, quasi-Monte Carlo, loss probabilities}

\section{Introduction}

Many stochastic problems in finance and risk management are high-dimensional with a univariate quantity of interest, say $\mu=\E(\Psi(\bX))$ for some integrable function $\Psi:\IR^d\rightarrow\IR$ and random vector $\bX\sim F_{\bX}$ for some $d$-dimensional distribution function $F_{\bX}$. Because $\mu$ rarely allows for an analytical expression,
the plain Monte Carlo (MC) estimator $\hmuMCn=(1/n)\sum_{i=1}^n\Psi(\bX_i)$ where $\bX_1,\dots,\bX_n\isim F_{\bX}$ is a popular choice for finding approximate solutions to such problems. It is unbiased and has an estimation error converging to zero at a rate independent of the dimension of the problem is often popular for finding approximate solutions to such problems.
The drawback of plain MC is the high computational cost it requires to obtain an estimate with a sufficiently small error. This issue is particularly severe for rare-event simulation, i.e., when $\P(|\Psi(\bX)|>0)$ is small, as then a typically very large number of samples is required to obtain non-zero observations and therefore an estimator with small variance. As such, plain MC is often combined with variance reduction techniques (VRTs), such as control variates (see, e.g., \cite{lavenbergwelch1981}) or stratified sampling (SS) (see, e.g., \cite{cochran2005}) to make the variance and thus the width of the estimate's confidence interval small.

Importance sampling (IS) is a VRT frequently applied to rare-event analysis in order to improve the reliability of MC estimators; see, e.g., \cite{kahnmarshall1953} and \cite{asmussenglynn2007}. The main idea of IS is to draw samples from a proposal distribution that puts more mass on the rare-event region of the sample space than the original distribution.
As the efficiency of IS depends heavily on the choice of the proposal distribution, finding a good proposal distribution is a crucial step in applying IS.  Unfortunately, there is no single best strategy known for finding a good proposal distribution that works in every situation since the nature of the rare event and what constitutes a good proposal distribution depends on the problem at hand; that is, on $\Psi$ and $F_{\bX}$.
As such, much of the existing work on IS in computational finance finds effective proposal distributions by exploiting the structure of specific problems: \cite{glassermanheidelbergershahabuddin1999} develop IS methods to price path-dependent options under multivariate normal models; \cite{glassermanheidelbergershahabuddin2000, glassermanheidelbergershahabuddin2002} estimate the Value-at-Risk of a portfolio consisting of stocks and options under a normal and $t$-distribution; \cite{sakhormannleydold2010} estimate tail probabilities of equity portfolios under generalized hyperbolic marginals with a $t$-copula assumption; \cite{glassermanli2005} estimate tail probabilities of credit portfolios under the Gaussian copula, \cite{bassamboojunejazeevi2008, chankroese2010} consider $t$-copula models.
As all these IS techniques are exploiting specific properties of the problem at hand, they can achieve substantial variance reduction but are typically specific techniques  not applicable to other problems without major modifications.

The contribution of this work is the development of theory and algorithms to apply IS for a wide range of problems by introducing a conditioning sampling step and optimally twisting the distribution of the conditioning variable. Let $T=T(\bX)$ be some univariate random variable, such as $\bbeta\T\bX$ for some (well chosen) $\bbeta\in\IR^d$, and assume sampling from $\bX \mid T$ is feasible. If $T$ has density $f$ (resp., $g$) under the original (resp., proposal) distribution (both distributions assumed to have the same support $\Omega_T$ for now), let
$$\hmuIS_n = (1/n) \sum_{i=1}^n \Psi(\bX_i) f(T_i)/g(T_i),\quad T_i\isim g,\quad \bX_i \isim F_{\bX\mid T}(\cdot \mid T = T_i),\quad i=1,\dots,n.$$
 If $T$ explains much of the variability of the output, so if $R^2 := \Var( \E(\Psi(\bX) \mid T) ) / \Var(\Psi(\bX))$ is large, we can choose $g$ optimally and make the rare event more likely by changing the distribution of $\bX$ through changing the distribution of the univariate $T$. Many high dimensional financial problems are of this nature; see, e.g., \cite{caflischrusselmorokoffowen1997, wangfang2003, wangsloan2005, wang2006}.

In order to analyze our estimator, we work with the semi-parametric model
$$\Psi(\bX)=m(T)+\eps_{\bX,T}$$
 for some (unknown) transformation  $T:\IR^d\rightarrow\IR$, where $m^{(k)}(t)=\E(\Psi(\bX)^k\mid T)$ for $k\in\IN$ and $\eps_{\bX,T}$ is a random error so that $\eps_{\bX,T}\mid T$ has mean 0 and variance $v^2(t)=\Var(\Psi(\bX)\mid T=t)$. We say that $\Psi(\bX)$ has a strong single index structure if $R^2$ is large (say, $R^2>0.9$), and the resulting estimator is referred to as Single Index IS (SIS) estimator.
We will show that the optimal proposal distribution for $g$ under SIS is proportional to $(m^{(2)})^{1/2}(t)f(t)$ resulting in an estimator with variance no larger than the plain MC estimator.
If the proposal distribution $g$ allows for a simple way to evaluate the quantile function $G_T^{-1}$ of $g$, we can further reduce the variance by applying equal stratification to the support of $T$, i.e., instead of sampling $T_1,\dots,T_n\isim g$, we can set $T_i = G_T^{-1}(U_i)$ where $U_k\isim \U( k/n, (k+1)/n)$ for $k = 0,\dots,n-1$ and $G_T^\i(u)=\inf\{t\in\IR: G_T(t)\geq u\}$ is the quantile function of $T$ under $g$. The resulting method is referred to as stratified SIS (SSIS). We also derive optimal variance expressions in this case and show that (S)SIS gives zero variance when $R^2=1$.
The derivation of these results along with some more notation and the connection between our methods and the IS and stratification techniques from \cite{arbenzcambouhofertlemieuxtaniguchi2018, glassermanheidelbergershahabuddin1999,  neddermeyer2011} can be found in Section~\ref{sec:SIS:var:opt}. There, we also briefly explain how our conditional sampling step reduces the effective dimension of the problem and therefore makes quasi-Monte Carlo (QMC) particularly attractive in our setting; in QMC, pseudo-random numbers (PRNs) are replaced by more homogeneously distributed quasi-random numbers (see, e.g., see \cite{niederreiter1978, lemieux2009, dickpillichshammer2010}).

Besides the choice of $g$, the performance of our procedure heavily depends on the choice of the transformation $T$, which must be chosen such that i) sampling from $\bX\mid T$ is feasible and ii) $T$ explains a lot of the variability of $\Psi(\bX)$, i.e., $R^2$ is as close to 1 as possible. The choice of the transformation is clearly not unique. In our numerical examples, we typically assume that $T$ is a linear function of $\bX$, whose coefficients can be estimated via the average derivative method of \cite{stoker1986}, the sliced inverse regression of \cite{li1991} or the semiparametric least-squares estimator of \cite{ichimura1993}. We remark that these methods do not require the form of the function $m(t)$ to be known.

As seen earlier, the optimal proposal densities involve a conditional moment function that is not known in practice. We propose to estimate this function using pilot-runs. The resulting point-wise approximation to the optimal density function can then be integrated and inverted numerically using the NINIGL algorithm developed in \cite{hormannleydold2003}. When this is too time-consuming, we suggest finding an approximately optimal $g$ in the same parametric family as $f$ (e.g., a location-scale transform of the original density). We detail this calibration stage, i.e., the process of estimating $T$, the optimal density and a way to sample from it, in Section~\ref{sec:calibration}.

In the numerical examples in Section~\ref{sec:numexp}, we demonstrate that our methods are applicable to a wide range of problems and achieve substantial variance reduction. After investigating a simple linear model example, we consider the problem of tail probability estimation in Gaussian and $t$-copula credit portfolio problems and show that our methods outperform those of \cite{glassermanli2005} and \cite{chankroese2010}.

As our formulation of (S)SIS does not assume a specific $\Psi$ or $F_{\bX}$, it is applicable to a wide range of problems and is efficient as long as the problem of interest has a strong enough single-index structure. It also adapts to the problem through the design of the one-dimensional transformation revealing the single-index structure and through the choice of the proposal distribution. Besides its applicability to a wide range of problems, our proposed method has the following advantages.
First, as it applies IS only to the univariate transformation variable, SIS is less susceptible to the dimensionality problem of IS, which is discussed in \cite{aubeck2003, katafygiotiszuev2008, schuelleretal2004}. This also simplifies the task of finding an optimal proposal distribution. Second, SIS has a dimension reduction feature, so it enhances the effectiveness of QMC sampling methods. Third, by applying IS to a transformation of the input random vector $\bX$, our proposal distribution amounts to changing the dependence structure of the problem under study, which can have a significant advantage over methods that only change the marginal distributions.

We conclude this paper in Section~\ref{sec:conclusion}.

\section{Variance analysis and optimal calibration for SIS and SSIS}
\label{sec:SIS:var:opt}

\subsection{Notations and definitions}
To fix notation, recall we estimate $\mu = \E(\Psi(\bX))$ via
\begin{align*}
\hmuSISn = (1/n) \sum_{i=1}^n \Psi(\bX_i) w(T_i),\quad T_i\isim g_T,\quad \bX_i \isim F_{\bX\mid T}(\cdot \mid T=T_i),\quad i=1,\dots,n,
\end{align*}
where $f_T$ and $g_T$ denote the original and proposal densities for $T$ with supports $\Omega_f=(t_{\inf},t_{\sup})$ (with possibly $t_{\inf},t_{\sup}\in\{\pm\infty\}$) and $\Omega_g$ and $w(t)=g_T(t)/f_T(t)$ is the IS weight function.

Furthermore, we model the output $\Psi(\bX)$ as $\Psi(\bX)=m(T)+\eps_{\bX,T}$, where
\begin{align*}
m^{(k)}(t)=\E(\Psi(\bX)^{k}\mid T),\quad \E(\eps_{\bX,T}\mid T)=0,\quad \var(\eps_{\bX,T}\mid T)=v^2(t)=\var(\Psi(\bX)\mid T).
\end{align*}
We already introduced the coefficient of determination $R^2 = \var(m(T))/\var(\Psi(\bX))$ (see, e.g., \cite{kvaalseth1985}) and said that $\Psi(\bX)$ is a strong single-index model if $R^2$ is large. This can be true for any model $\Psi(\bX)$, as we allow $\eps_{\bX,T}$ to depend on $\bX$. However, a pure single index model is a situation where $\eps_{\bX,T}=\eps_T$ only depends on $\bX$ through $T$. In that case, it is easy to see that $\E(\Psi(\bX)\mid T)=m(T)$, so that overall the random variable $\Psi(\bX)$ depends on $\bX$ only through $T$. However, we do not impose the assumption of a pure single index model. Readers are referred to \cite{powellstockstoker1989}, \cite{hardlehallichimura1993} and \cite{ichimura1993} for more information on single-index models.

Based on the representation of $\Psi(\bX)$ and using the law of total variance, we can write
\begin{align}\label{eq:SIS:vardecomp}
\var(\Psi(\bX))= \var(m(T)) + \E(v^2(T)) = \var(m(T)) + \var(\eps_{\bX,T}),
\end{align}
since $\E(v^2(T))=\var(\eps_{\bX,T})-\Var(\E(\eps_{\bX,T}\mid T)) = \var(\eps_{\bX,T})$.  We see that~\eqref{eq:SIS:vardecomp} decomposes the variance of $\Psi(\bX)$ into two pieces: the one of the (random) systematic part, $m(T)$ and the unsystematic error $\eps_{\bX,T}$ of the model. Note that~\eqref{eq:SIS:vardecomp} holds irrespective of wether we have a pure single index model or not.

In addition to applying IS on $T$, we also propose to use stratification on $T$ to further reduce the variance; it will turn out that this essentially ``stratifies away'' $\Var(m(T))$, the variance of the systematic part of the model. More precisely, let $\Omega_{f}=(t_{\inf},t_{\sup})$ where possibly $t_{\inf}=-\infty$ and $t_{\sup}=\infty$. The SSIS scheme splits $\Omega_{f}$ into $n$ strata of equal probability under $g$ and draws one sample of $T$ from each stratum. Our estimator becomes
\begin{align*}
\hmuSSISn = (1/n) \sum_{i=1}^n \Psi(\bX_i) w(T_i),\quad T_i=G_T^\i(U_i),\quad U_i\sim \U( (i-1)/n, i/n),
\end{align*}
and, as before, $\bX_i \isim F_{\bX\mid T}(\cdot \mid T=T_i)$ for $i=1,\dots,n$.

For our variance analysis below, it is useful to find an expression for $\var(\hmuMCn)$. Note that the conditional moment functions $m^{(k)}$ do not depend on wether we sample from $f_T$ or $g_T$. From~\eqref{eq:SIS:vardecomp} and the fact that  $\var(m(T))=\E(m(T))^2- \mu^2$ as well as $\E(v^2(T)) = \E(m^{(2)}(T))-\E(m(T))^2$, we find

\begin{align}\label{eq:var:mu:MC}
n \var(\hmuMCn) = \var(m(T))+\E(v^2(T)) = \E(m^{(2)}(T))-\mu^2.
\end{align}

As should be clear from the form of our estimators, their bias depends on the support $\Omega_{g}$ of $g_T$. We define
\begin{align*}%
\muSIS = \int_{\Omega_g} m(t)f_T(t)\;\rd t, \quad \sigSISsq=\int_{\Omega_g} m^{(2)}(t) \frac{f_T^2(t)}{g_T(t)}\;\rd t  - \muSIS^2, \quad \sigSSISsq=\int_{\Omega_g} v^{2}(t) \frac{f_T^2(t)}{g_T(t)}\; \rd t.
\end{align*}
Notice that $\muSIS$ depends on $g_T$ through the region $\Omega_{g}$. The SIS and SSIS estimators are unbiased only if $g_T$ is such that $g_T(t)>0$ whenever $m(t)f_T(t)>0$, but we do not impose this unbiasedness assumption on $g_T$.

\subsection{Optimal densities}

We are now able to derive properties of the (S)SIS estimators and derive the optimal (variance-minimizing) proposal distribution of $g_T$; see the appendix for the proofs. As the objective of our IS techniques is variance reduction, we call the practice of setting $g_T$ to its optimal density or their approximation as \emph{optimal calibration}, and the resulting methods $\operatorname{SIS}^*$ and  $\operatorname{SSIS}^*$.

\begin{proposition}[Variance-optimal SIS]\label{prop:IS}
We have $\E(\hmuSISn)=\muSIS$ and $\var(\hmuSISn)=\sigSISsq/n$. If $\E_g(m^2(T)w^2(T)) < \infty$, then $\sqrt{n} (\hmuSISn-\muSIS) \darrow \N(0,\sigSISsq)$ as $n\rightarrow\infty$.

Suppose that $\Psi(\bx) \geq 0$ or $\Psi(\bx) \leq 0 $ for all $\bx \in \Omega_{\bX}$. The density $g_T$ that gives an unbiased SIS estimator with the smallest variance is
	\begin{align}\label{eq:IS:gopt}
	\gTopt(t)= c^{-1} \sqrt{m^{(2)}(t)} f_{T}(t),
	\quad t \in (t_{\inf},t_{\sup}),\quad c = \int_{t_{\inf}}^{t_{\sup}} \sqrt{m^{(2)}(t)} f_{T}(t)\; \rd t.
	\end{align}
The variance of the optimal SIS estimator, denoted by $\hmuSISoptn$, is $\var(\hmuSISoptn)=(c^2 - \mu^2)/n$.
\end{proposition}

\begin{remark}\label{rem:optIS}
\begin{enumerate}
\item Proposition~\ref{prop:IS} implies that using optimal SIS gives variance no larger than MC. Indeed, by Jensen's inequality, $n\var(\hmuSISoptn) \leq \E(m^{(2)}(T))-\mu^2$, which is equal to  $\var(\hmuMCn)$ using~\eqref{eq:var:mu:MC}. This inequality holds as an equality only when $m^{(2)}(t)$ is constant for all $t \in \Omega_{T}$.
\item If $R^2 = 1$ (corresponding to the strongest possible single index structure), then $\var(\hmuSISoptn)=0$:  SIS provides a zero-variance estimator if $m^{(2)}(t)=(m(t))^2$ for all $t$, which is equivalent to having $v^2(t)=0$ for all $t$, or equivalently, to having $\E(v^2(T))=0$ since $v^2(t) \ge 0$ for all $t$. This is the same as asking $\var(m(T))/\var(\Psi(\bX))=R^2=1$. This is why choosing a function $T$ such that the model is an as good fit as possible is important for the SIS method to achieve significant variance reduction.
\end{enumerate}
\end{remark}

The following proposition gives the properties of the SSIS estimator and the optimal (variance-minimizing) proposal distribution of $g_T$. Its proof is in the appendix.
\begin{proposition}[Variance-optimal SSIS]\label{prop:SIS}
It holds that $\E(\hmuSSISn)=\muSIS$ and, for large enough $n$, $\var(\hmuSSISn)=\sigSISsq/n + o(1/n)$. If $\E_g \left(\left|m(T)w(T)\right|^{2+\delta}\right) < \infty$ for some $\delta>0$, $\hmuSSISn$ is asymptotically normal as $\sqrt{n} (\hmuSSISn-\muSIS) \darrow \N(0, \sigSISsq)$ for $n\rightarrow\infty$.
Suppose that $\Psi(\bx) \geq 0$ or $\Psi(\bx) \leq 0$ for all $\bx \in \Omega_{\bX}$ and that $\P_f(v^2(T)=0, \; m(T) \neq 0)=0$.
The density $g_T$ that gives an unbiased SSIS estimator with the smallest variance is
	\begin{align}\label{eq:SIS:gopt}
	\gTopts(t)=c^{-1}v(t) f_{T}(t), \quad t \in (t_{\inf},t_{\sup}), \quad c=\int_{t_{\inf}}^{t_{\sup}} v(t) f_{T}(t) \;\rd t.
	\end{align}
The variance of the optimal SSIS estimator $\hmuSSISoptn$ is $\var(\hmuSSISoptn)=c^2/n+o(1/n)$. If $\P_f(v^2(T)=0, \; m(T) \neq 0)>0$, then $\hmuSSISoptn$ is biased.
\end{proposition}

\begin{remark}
\begin{enumerate}
\item Proposition~\ref{prop:SIS} implies that using optimal SSIS gives asymptotically a variance no larger than MC. Indeed, Jensen's inequality implies that we have $\var(\hmuSSISoptn) \leq (1/n) \E(v^2(T)) + o(1/n)$ with equality only if $v(t)$ is constant for all $t \in \Omega_T$. From~\eqref{eq:var:mu:MC} (and ignoring the $o(1/n)$ term), this means $\var(\hmuSSISoptn) \le \var(\hmuMCn)$, with equality only if $v(t)$ is constant for all $t \in \Omega_T$ and $\var(m(T))=0$, which is unlikely to be the case since $m(T)$ has been chosen specifically such that $R^2\approx 1$.
\item If $R^2 = 1$ (strongest possible single index structure), then $\var(\hmuSSISoptn) =0$, since $\var(\hmuSSISn)=0$ iff $m^{(2)}(t)=(m(t))^2$ for all $t$, or equivalently $v^2(t)=0$ for all $t$ and thus $\E(v^2(T))=0$, which means $R^2=1$.
\item Unless $m(t)=0$, SSIS achieves variance reduction compared to SIS, as $\var(\hmuSSISn) \leq \var(\hmuSISn)$ for the same choice of $g_T$. This in turn implies that $\var(\hmuSSISoptn) \leq \var(\hmuSSISoptn)$.\label{eq:SISbetterIS}
The proposal densities $\gTopt$ and $\gTopts$ defined in~\eqref{eq:IS:gopt} and~\eqref{eq:SIS:gopt} give estimators with smallest variance if $\Psi(\bx) \lesseqgtr 0$ for all $\bx \in \Omega$, which holds for many applications in finance (e.g., when $\Psi$ is an indicator and thus $\mu$ a probability or when $\Psi$ is the payoff of an option).
If $\Psi$ takes both positive and negative values, $m(t)$ could be 0 for some values of $t$. We can then improve the optimal calibration by setting $g_T(t)=0$ whenever $m(t)=0$. Since it is generally unknown and hard to estimate which values of $t$ give $m(t)=0$, this improvement may not be implementable.
\item The expression for $\var(\hmuSSISoptn)$ implies that $\SSIS^*$ ``stratifies away'' the variance captured by the systematic part $m(T)$ of the single-index model, so the variance of the  $\SSIS^*$ estimator comes only from the error term $\eps_{\bX,T}$ via $v(t)$. If $g_T$ is not chosen optimally, then $\var(\hmuSSISn) = \sigSISsq/n + o(1/n)$ shows that we still make $\var(m(T))$ vanish by using stratification, but the contribution from $v^2(T)$ might be amplified (compared to how it contributes to the MC estimator's variance) if we do not choose a good proposal density. Irrespective of the choice of $g_T$ it is true that the stronger the fit of the single index model, the better (S)SIS works.
\item These results show that as long as the problem at hand has a strong single-index structure and sampling from $T$ and $\bX\mid T$ is feasible, SIS and SSIS can be applied and should give large variance reduction. As those conditions do not assume a specific form for $\Psi$ or for the distribution of $\bX$, SIS and SSIS are applicable to a wide range of problems.
\end{enumerate}
\end{remark}

Proposition~\ref{prop:SIS} asserts the asymptotic normality of the SSIS estimator. In order to construct a confidence interval from this estimator, we must estimate $\sigSSISsq$. We take an approach similar to the one by \cite{wangbrowncailevine2008} where the first-order difference of samples are taken to remove the effect of the mean function. Its proof is in the appendix.

\begin{proposition}[Estimation of $\sigSSISsq$]\label{prop:SISvar:est}
Let $G_T$ be the distribution function corresponding to $g_T$.
	If $G^{-1}_T$, $m$ and $v^2$ are continuously differentiable over the domain of $T$ under the proposal distribution, then
	\begin{align*}
	\hsigSSISsq = \frac{1}{2(n-1)}\sum\limits_{i=1}^{n-1} r_i^2 w^2(T_i)
	\end{align*}
	is a consistent estimator of $\sigSSISsq$, where $r_i= \Psi(\bX_{i+1}) - \Psi(\bX_i)$ for $i=1,\ldots,n-1$.
\end{proposition}

Proposition~\ref{prop:SISvar:est} assumes that $G^{-1}_T$ is continuously differentiable which requires that $g_T(t)>0$ on the support of $T$ under the proposal distribution. This does not hold if there exist intervals where $g_T(t)=0$. In such a situation, we propose to divide the support of $T$ into disjoint intervals with $g_T(t)>0$ then apply Proposition~\ref{prop:SISvar:est} separately to each interval and combine them to obtain $\hsigSSISsq$.

\subsection{Connection to other IS and SS techniques}

In this subsection, we explain some connections of our proposed methods to other IS and SS techniques.

Suppose that $\bX \sim \N_d(\bzero, I_d)$. A popular strategy for constructing a proposal distribution under the multivariate normal (MVN) model is to shift its mean vector of $\bX$, that is, letting $\bX \sim \N_d(\bm{\eta}, I_d)$ under the IS distribution for some $\bm{0} \neq \bm{\eta} \in \IR^d$.
The following proposition states that this type of IS can be achieved within our SIS framework by using $T(\bX)= \btheta\T\bX$ where $\btheta$ is the normalized version of $\bm{\eta}$. Based on Proposition~\ref{prop:IS} and Remark~\ref{rem:optIS}, this result thus implies that this popular mean-shifting strategy for MVN models works well if the problem has a strong linear single-index structure based on the specific choice of shift vector $\bm{\eta}$.

\begin{proposition}[SIS in MVN models]\label{prop:meanshift}
Let $\bX \sim \N_d(\bm{0}, I_d)$ under the original distribution. Fix $\bm{0} \neq \bbeta \in \IR^d$ with $\bbeta\T\bbeta=1$.
Consider SIS with $T(\bX)= \bbeta\T\bX$. If $g_T$ is the density of $\N(c,\sigma^2)$, then $\bX \sim \N_d(c\bbeta, I_d + (\sigma^2-1)\bbeta\bbeta\T)$ in the IS scheme.
\end{proposition}

Proposition~\ref{prop:meanshift} implies that $\bX \sim \N_d(c\bbeta, I_d)$ if $g_T$ is chosen as the $\N(c, 1)$ density (where we recall that the original distribution $f_T$ is $\N(0,1)$), so that the previously mentioned mean-shifting strategy is a special case of IS (namely, by merely shifting the mean of $T$ instead of applying $\operatorname{SIS}^*$). If $\var(T)\ne 1$ under $g$, the dependence structure of the components in $\bX$ does change in the IS scheme.

The stratification technique proposed in \cite{glassermanheidelbergershahabuddin1999} is applied by using the normalized shift vector as the stratification direction and can also be achieved within our SIS framework using the same function $T$ and proposal distribution as in Proposition~\ref{prop:meanshift}. The combination of IS and SS is not motivated as in \cite{glassermanheidelbergershahabuddin1999}. In the latter reference, IS and SS are used to remove the variability due to the linear and the quadratic part, respectively, of $\Psi(\bX)$. In SSIS, SS is used to eliminate $\var(m(T))$, the variance captured by the systematic part of the single-index model, and then IS is used to minimize the variance contribution from $\eps_{\bX,T}$.

It is easy to see that the NPIS method proposed by \cite{neddermeyer2011} with $u=1$ (where $u$ is defined as in \cite{neddermeyer2011}) is closely connected to SIS with $T(\bX)=X_1$. It is proposed to choose $g_T(t)= m(t)f_T(t)/\mu$ in \cite{neddermeyer2011}, but by Proposition~\ref{prop:IS}, choosing $\gTopt$ defined in \eqref{eq:IS:gopt} gives an IS estimator with a smaller variance.

SIS also generalizes the IS method in \cite{arbenzcambouhofertlemieuxtaniguchi2018} in two ways. First, SIS generalizes the form of the transformation function $T$, that is, it does not assume any specific form of $T$, while the IS method in \cite{arbenzcambouhofertlemieuxtaniguchi2018} assumes that $T(\bX) = \max \{F_1(X_1),\dots,F_d(X_d)\}$, where $F_1,\dots,F_d$ are the marginal distribution functions of $\bX$. Secondly, SIS generalizes the form of the proposed density of the transformed variable, whereas the proposal density $g_T$ for the IS method in \cite{arbenzcambouhofertlemieuxtaniguchi2018}  has the form
	\begin{align*}
	g_{T}(t) = \sum_{k=1}^{M} q_k f_{T}(t \mid T>\lambda_k)
	=\sum_{k=1}^{M} q_k \frac{f_{T}(t)I_{\{t>\lambda_k\}}}{1-F_T(\lambda_k)},
	\end{align*}
for some $M \geq 1$, $t_{\inf}=\lambda_1 < \cdots < \lambda_{M}$, and $q_1,\ldots,q_M \geq 0$ such that $\sum_{k=1}^{M} q_k=1$.

The single-index structure we exploit to design our SIS and SSIS schemes is strongly related to the idea of conditional MC. In both cases, the goal is to identify a function $T$ of $\bX$ that explains much of the variability of $\Psi(\bX)$. However, with conditional MC one typically also chooses $T$ so that $m(t)=\E(\Psi(\bX)\mid T(\bX)=t)$ is known, and then estimates $\mu$ by the sample mean of the $m(T_i)$, $i=1,\ldots,n$. In our case, we do not assume or need this conditional expectation to be known in closed-form. This means we typically do not completely get rid of the $\var(m(T))$ term in~\eqref{eq:SIS:vardecomp}, but we aim to reduce it via IS; if SSIS is applied optimally, we actually do make $\var(m(T))$ vanish.

\subsection{Single-Index Importance Sampling and QMC}\label{sec:qmc:IS}

As mentioned in the introduction, further variance reduction can be achieved by performing the simulation based on quasi-random numbers (QRNs) instead of PRNs. Suppose we are given a sampling algorithm $\phi : [0,1)^{d+k} \rightarrow \IR^d$ for some $k \geq 0$ such that
$\phi(\bU) \sim f_{\bX}$ for $\bU \sim \U[0,1)^{d+k}$. For instance, when $\bX\sim \N_d(\bmu, \Sigma)$, then $k=0$ and $\phi(\bu)=\bmu+C (\Phi^{-1}(u_1),\ldots,\Phi^{-1}(u_d))\T$, where the matrix $C$ is such that $CC\T = \Sigma$ and $\Phi(x)=\int_{\infty}^x (2\pi)^{-0.5} \exp(-t^2/2)\;\rd t$ is the distribution function of the standard normal distribution. For a discussion of what the function $\phi$ is in a more general context, where $\bX$ has a dependence structure modelled by a copula other than the Gaussian copula, we refer to \cite{cambouhofertlemieux2016}.
With $\phi$ at hand, we can write $\hmuMCn=(1/n)\sum_{i=1}^n \Psi(\phi(\bU_i))$ where $\bU_i\isim\U(0,1)^{d+k}$. With QMC, we replace the $\bU_i$ with deterministic vectors $\bv_i\in[0,1)^{d+k}$ that fill the unit hypercube more evenly.  A number of constructions for such points have been proposed (see e.g., \cite[Ch.~5]{lemieux2009}); we use the Sobol' sequence of  \cite{sobol1967} for our numerical examples later on.
In order to obtain an easy-to-compute error bound, we apply a random digital shift to the $\bv_i$ to obtain multiple independent and identically distributed realizations of the randomized QMC (RQMC) estimator. Based on the digitally-shifted RQMC estimates, we can compute a probabilistic error bound in the form of a confidence interval.

It is widely accepted that the performance of QMC is largely influenced by the effective dimension of the problem, a concept first introduced in \cite{caflischrusselmorokoffowen1997}. More precisely, QMC works significantly better than plain MC if the problem has a low effective dimension; see also \cite{wangfang2003, wangsloan2005, wang2006}. One notion of effective dimension is the truncation dimension; see \cite{wangsloan2005}. Essentially, a problem has a low truncation dimension when only a small number of leading input variables are important. Recall that $\bX$ is sampled indirectly in SIS, that is, $T$ is generated first then $\bX$ is drawn from $F_{\bX | T}$. Assuming $T$ is generated using the inversion method and via the first coordinate $u_1$ of $\bu \in [0,1)^{k+d}$, the indirect sampling step of SIS transforms the problem in such a way that the first input variable accounts for $R^2\cdot100\%$ of the variance of $\Psi(\bX)$, where $R^2=\var(m(T))/\var(\Psi(\bX))$. That is, the problem has a truncation dimension of 1 in proportion $R^2$ under SIS.
Therefore, if the fit of the single-index model is good, say $R^2>0.9$, the indirect sampling step via $T$ serves as a dimension reduction technique and enhances the efficiency of QMC.

\section{Calibration stage in practice}\label{sec:calibration}

As mentioned in the introduction, we must estimate the optimal transformation function $T=T(\bX)$ and construct an approximation $\hgTopt$ for the optimal density $\gTopt$ before applying (S)SIS. We call the stage in which these two tasks are performed the \emph{calibration stage}. Furthermore, the calibrations in \eqref{eq:IS:gopt} and~\eqref{eq:SIS:gopt} require the knowledge of the conditional mean function and variance function, respectively. As these are rarely known in practice, they must be estimated in the calibration stage as well.

\subsection{Estimating the optimal transformation $T$}
In what follows, we assume that $T$ is a linear function of the components in $\bX$, i.e., $T=\bbeta\T\bX$ for some $\bbeta\in\IR^d$; note that if $\bX$ is multivariate normal, then $T$ is univariate normal and sampling from $\bX\mid T$ is straightforward. To find $\bbeta$ that maximizes $R^2$, we use the average derivative method of \cite{stoker1986}, which essentially allows us to estimate $\bbeta$ as if we met the assumptions of a linear regression. That is, we sample independent realizations $\Psi(\bX_i)=:\Psi_i$ for $i=1,\dots,n_1$ (say, $n_1=1000$) and compute the sample covariance matrix $\Sigma_{\bX,\bX}$ as well as the sample cross covariance of $\bX_1,\dots,\bX_{n_1}$ and $(\Psi_1,\dots,\Psi_{n_1})$, say $\Sigma_{\bX,\Psi}$ to obtain

$$\hat{\bbeta} = \Sigma_{\bX,\bX}^{-1} \Sigma_{\bX,\Psi}.$$

In some applications, we may use only a subset of the components in $\bX$; in later examples, for instance, we only use the systematic risk factors in a credit model to build our transformation $T$. Sometimes one may even not need to estimate $\bbeta$, for instance, if it is clear that the $d$ components in $\bX$ are equally important, one can simply set $\bbeta = (1/\sqrt{d},\dots,1/\sqrt{d})$; see Section~\ref{section:credit:gauss} for an example.

\subsection{Finding the optimal density}

The calibration in~\eqref{eq:IS:gopt} requires the knowledge of the conditional second moment function $m^{(2)}(t) = \E(\Psi^2(\bX)  \mid T=t)$ for all $t\in \Omega_T$, which, of course, is not known; similarly, the conditional variance function $v^2$ required for the calibration in~\eqref{eq:SIS:gopt} is not known either. We now describe how to calibrate~\eqref{eq:IS:gopt} in practice; the calibration of~\eqref{eq:SIS:gopt} can be done similarly.

Our first ingredient is the construction of an estimate of $m^{(2)}(t) = \E(\Psi^2(\bX)  \mid T=t)$ for all $t\in \Omega_T$; we suggest using plain MC for this purpose. To this end, let $t_1<\dots<t_M$ be knots at which the function $m^{(2)}$ is to be estimated (e.g., $M=20$ equally spaced points in the relevant range). Choose some small pilot sample size $\npil$ (for example, 5\% of the total sample size $n$). For each $t_j$, sample $\npil$-many realizations from $\bX \mid T=t_j$ and estimate $m^{(2)}(t_j)$ by its empirical equivalent for $j=1,\dots,M$. Then utilize smoothing splines (see, for example, \cite{reinsch1967}) and only those $t_j$ associated with a positive estimate to construct an estimate $\hat{m}^{(2)}$ for all $t\in \Omega_T$; for those $t$ where $\hat{m}^{(2)}(t)\leq 0$, one can either leave them as $\hat{m}^{(2)}(t)=0$ (which may lead to bias as discussed below) or set $\hat{m}^{(2)}$ to be some positive function (e.g., the error function) that resembles the lower tail of $\eps$.

Having constructed an estimate for $m^{(2)}$,  we can set $\hgTopt\propto \sqrt{\hat{m}^{(2)}(t)}f(t)$ for $t\in\IR$. However, $\hgTopt$ rarely belongs to known parametric families of distributions that are easily sampled from. One can use numerical techniques such as the NINIGL algorithm to approximate the quantile function of a distribution given its unnormalized density; see \cite{hormannleydold2003}. This approach, however, has three drawbacks: i) sampling from a numerically constructed density is time-consuming and can be prone to numerical problems; ii) the normalizing constant needs to be estimated, and iii) bias can occur when $\hgTopts$ does not have the same support as $\gTopts$, which in turn happens when $\hat{m}^{(2)}(t)=0$ even though $m^{(2)}(t)\not=0$ for some set $D$ with $\int_D f(t)\;\rd t>0$.

The third drawback can be alleviated if we can define $\hat{m}^{(2)}(t)$ to be positive whenever $m^{(2)}(t)$ is (for example, by assuming some lower and upper tail behaviour). Furthermore, recall from Proposition~\ref{prop:SIS} that \eqref{eq:SIS:gopt} gives a biased estimator if $\P_f(v^2(T)=0, \; m(T) \neq 0)>0$ which in some cases can be debiased. For instance, if $v(t)>0$ for all $t\in\Omega_t$, but the estimated $\hat{v}(t)=0$ for $t\geq t_{\max}$ for some $t_{\max}\in\IR$ and $m(t)=c$ for some constant $c$ for $t\geq t_{\max}$ (for instance, if $\mu$ is a probability, then typically $m(t)=c=1$ for $t\geq t_{\max}$). If $\hmuSSISn$ is constructed using $\hat{v}$, we find
\begin{align*}
\E(\hmuSSISn) &= \int_{-\infty}^{t_{\max}} m(t)f_T(t)\;\rd t = \mu - \int_{t_{\max}}^{\infty} m(t)f_T(t)\;\rd t \approx \mu - c\P_f(T>t_{\max});
\end{align*}
$\hmuSSISn$ can therefore be debiased by adding $c\P_f(T>t_{\max})$.

The second drawback can be addressed by using weighted IS (so that the normalizing constant cancels out); see \cite[Section~4.5]{lemieux2009}. Alternatively, the normalizing constant can be estimated as follows: Let $\hgToptu(t)= \sqrt{\hat{m}^{(2)}}g(t)$ denote the unnormalized density, and $T_1,\dots,T_n\isim \hgTopt$ (obtained, for instance, using the NINIGL algorithm). Now construct an estimate of the density of $T_1,\dots,T_n$, such as the kernel density estimator, and denote this estimated density by $\hat{h}$; note that $\hat{h}$ is normalized and that each of $\hat{h}(T_i)/\hgToptu(T_i)$ for $i=1,\dots,n$ is an estimator for the normalizing constant. As such, we suggest using the sample median of $\{\hat{h}(T_1)/\hgToptu(T_1),\dots,\hat{h}(T_n)/\hgToptu(T_n) \}$ as an estimator for the normalizing constant.

The first drawback, that is, the construction of an approximation to the quantile function of $\hgTopt$ being both slow and potentially prone to numerical problems, is most severe. Below, we propose an alternative method, namely by setting $\hgTopt(t) = 1/\sigma f( (t-k)/\sigma)$ for carefully chosen $k\in\IR$ and $\sigma>0$. In other words, we suggest using a location-scale transform of the original density as proposal density and will therefore call this method $\SIS^{c,\sigma}$. While this procedure does require estimation of $k$ and $\sigma$, it does not suffer from any of the three aforementioned problems: i) if we can sample from $f$, we can also sample from $f( (t-k)/\sigma)/\sigma$; ii) there is no normalizing constant or density to be estimated; iii) $f$ and $f( (t-k)/\sigma)/\sigma$ have the same support, so that the resulting estimator is unbiased.

The idea behind using a location-scale transform arises from the observation that in many practical examples (as will be seen later) the optimal density has roughly the same shape as the original density. As such, we try to find $k$ and $\sigma$ so that $1/\sigma f( (t-k)/\sigma)$ is approximately $\gTopt(t)$. Denote again by $\hgToptu(t)= \sqrt{\hat{m}^{(2)}(t)}f(t)$ the unnormalized, estimated optimal density and assume that the mode of $f_T$ is at zero (otherwise, shift accordingly).  Now find $k^* = \argmax_t \hgToptu(t)$ numerically; this makes sure that the theoretical and approximated densities have (roughly) the same mode, thereby both sample from the ``important region''. Having estimated $k^*$, the next step is to compute $\sigma$ such that it minimizes the variance of the resulting estimator. More precisely, given a sample $T_1,\dots,T_{\npil}$ from $f$,
we can estimate the variance of the estimator for a given $\sigma$ as follows: Set $\tilde{T}_i=k^* + \sigma T_i$ and $w_i =\frac{f(\tilde{T}_i)}{f( (\tilde{T}_i-k^*)/\sigma)/\sigma}$ and sample $\bX_i \mid \tilde{T}_i$ for $i=1,\dots,\npil$. The second moment of the IS estimator (written as a function of the scale $\sigma$) is then
\begin{align}\label{eq:Vsigma}
 V(\sigma) =  \sum_{i=1}^{\npil} \Psi(\bX_i) w_i^2, \quad \sigma>0.
 \end{align}
We can now solve $\sigma^* = \argmin_{\sigma>0} V(\sigma)$ numerically. Note that due to the nature of a location-scale transform, we only need to sample $T_1,\dots,T_{\npil}$ once. Intuitively, $k^*$ shifts the density to the important region, while $\sigma^*$ scales it appropriately. If computing $V(\sigma)$ is very time consuming (for example, when the sampling of $\bX \mid T$ is complicated), one can set $\sigma^*=1$; the resulting method is then called $\SIS^{\mu}$ instead of $\SIS^{\mu,\sigma}$.

\begin{algorithm}[Calibration and estimation stage for estimating $\mu$ via $\SIS^{\mu,\sigma}$]\label{alg:SIS:full}
Given knots $t_1,\dots,t_{\npil}$, a total pilot budget $\ntot$ and knot-sample size $\nknot$, target sample size $n$, estimate $\mu$ via:
\begin{enumerate}

\item \emph{Estimation of the direction vector.}
\begin{enumerate}
\item Sample $\bX_i\isim f_{\bX}$ for $i=1,\dots,\npil$ and compute $\Psi(\bX_i)=:\Psi_i$, $i=1,\dots,\npil$.
\item Compute  $\Sigma_{\bX,\bX}$ and $\Sigma_{\bX,\mathbf{\Psi}}$and set $\hat{\bbeta} = \Sigma_{\bX,\bX}^{-1} \Sigma_{\bX,\Psi}$.
\end{enumerate}

\item \emph{Estimation of $c^*$ and $\sigma^*$.}
\begin{enumerate}
\item\label{step:condsampling1}  For each $k=1,\dots,\npil$, sample $\bX_{j,k}\isim f_{\bX\mid T}(\cdot \mid t_k)$ for $j=1,\dots,\nknot$.
\item Utilize smoothing splines\footnote{In the case when $\Psi$ is an indicator, use a logistic regression (available, for instance, via the \R\ function \texttt{glm())} with $\nknot = 1$ instead.} through  $(t_k, (1/\nknot)\sum_{j=1}^{\nknot} \Psi(\bX_{j,k})^2)$, $k=1,\dots,\npil$, to construct an estimate for $\hat{m}^{(2)}(t)$ for $t\in\IR$.
\item Find $c^* = \argmax_t \sqrt{\hat{m}^{(2)}(t)} f(t)$ numerically.
\item\label{step:condsampling2} Sample $T_1,\dots,T_{\npil}\isim f$ and find $\sigma^* = \argmin_{\sigma>0} V(\sigma)$ with the function $V$ from~\eqref{eq:Vsigma} numerically.
\end{enumerate}

\item \emph{Estimation of $\mu$.}
\begin{enumerate}
\item Sample $T_1',\dots,T_n'\isim f$, set $T_i = c^* + \sigma^* T_i$ and compute $w_i =\frac{f(T_i)}{ f( (T_i-c^*)/\sigma^*)/\sigma^*}$ for $i=1,\dots,n$.
\item\label{step:condsampling3}  Sample $\bX_i\isim f_{\bX\mid T}(\cdot \mid T_i)$ for $i=1,\dots,n$.
\item Return $\hmuSISn = (1/n) \sum_{i=1}^n \Psi(\bX_i)w_i$.
\end{enumerate}

\end{enumerate}
\end{algorithm}

\begin{remark}
\begin{enumerate}
\item Algorithm~\ref{alg:SIS:full} can be easily adapted to accommodate quasi-random numbers and stratification, as will be discussed in the next section.
\item The effort for the conditional sampling needed in Steps~\ref{step:condsampling1}, \ref{step:condsampling2} and~\ref{step:condsampling3} is problem specific -- for some problems, samples of $\bX \mid T=t_1$ can be easily transformed to samples from $\bX \mid T=t_2$ for $t_1\not=t_2$, making these steps very fast; in some other problems, the conditional sampling is more involved.
\item Our proposed SIS method can also be combined with other VRTs. For instance, in Section~\ref{section:credit:t}, we combine conditional MC (CMC) and SIS to estimate loss probabilities of a credit portfolio whose dependence is governed by a $t$-copula.
\end{enumerate}
\end{remark}

\section{Numerical Experiments}
\label{sec:numexp}

In this section, we perform an extensive numerical study to demonstrate the effectiveness of our proposed methods. We start with a simplistic linear model example, in which case calibration of the optimal densities can be done easily. This allows us to investigate the effect of replacing $\gTopt$ by $\hgTopt$. In Section~\ref{section:credit:gauss}, we apply our SIS and SSIS schemes to a credit portfolio problem under the Gaussian copula model studied by \cite{glassermanli2005}. The same financial problem but this time using a more complicated $t$-copula model is studied in Section~\ref{section:credit:t}.   All computations were carried out in \R; see \cite{Rsystem}.
\subsection{Linear Model Example}\label{section:linearmodel}

Let $L = \alpha  T + \eps_T$ where $T\sim \N(0,1)$, $\eps_T \mid T \sim \N(0, s^2)$ and $\alpha^2+s^2 = 1$. $L$ has a single index structure when $\alpha^2\approx 1$ since $R^2 = \var(m(T))/\var(L)=\var(\alpha T)=\alpha^2$.

Assume interest lies in estimating the probability $p_l=\P(L>l)=\bar{\Phi}(l)=\E(\I_{\{L>l\}})$ for some large $l$; note that we can approximate the true value of $p_l$ efficiently with high precision since $L\sim\N(0,1)$. Furthermore it is easily seen that $p_l(t) = \P(L>l \mid T = t) =  \bar{\Phi}\left( (l-\alpha t) /s\right)$ for $l,t\in\IR$. Since the integrand $\Psi$ in this setting is an indicator, we find from Proposition~\ref{prop:IS} that $\gTopt(t) \propto \sqrt{p_l(t)}f_T(t)$.

Unlike in this simplistic setting, $p_l(t)$ for $t\in\IR$ is unknown in practice as discussed in Section~\ref{sec:calibration}; thus, this setting serves as an excellent example to also compare whether approximating $\gTopt$ by $\hgTopt$ has a significant effect on the accuracy of the estimators. Sampling from the true optimal densities is performed using the \R\ package \texttt{Runuran} of \cite{leydoldhoermann2020}. We consider the methods $\SIS^{**}$ (constructed using known $p_l(t)$), $\SIS^{*}$ (approximated $p_l(t)$ and NINGL), $\SIS^{\mu}$ and $\SIS^{\mu,\sigma}$

For the two settings of $\alpha^2 \in\{0.7, 0.99\}$ (corresponding to a weaker and stronger single index structure), we estimate $p_l$ for $l\in\{3,4,\dots,7\}$ using the five aforementioned methods. For each value of $l$, the optimal density is calibrated separately. In all examples, we use a sample size of $n=10^6$ and a pilot sample size of $5\times 10^4$. We repeat the experiment 200 times.

Figure~\ref{fig:lm:densdur} displays on the left the optimally calibrated and approximated IS densities. The true optimal density is bell shaped, so it is well approximated by a normal density. It can be confirmed from the plot that in this case, all IS densities seem to cover the important range. The right of Figure~\ref{fig:lm:densdur} displays a boxplot of run-times needed to estimate $p_l$; note that the run-time does not depend on $\alpha$ or $l$. This plot, however, should be interpreted with caution as it highly depends on how the pilot runs are implemented.

\begin{figure}
\centering
\includegraphics[width=0.48\textwidth]{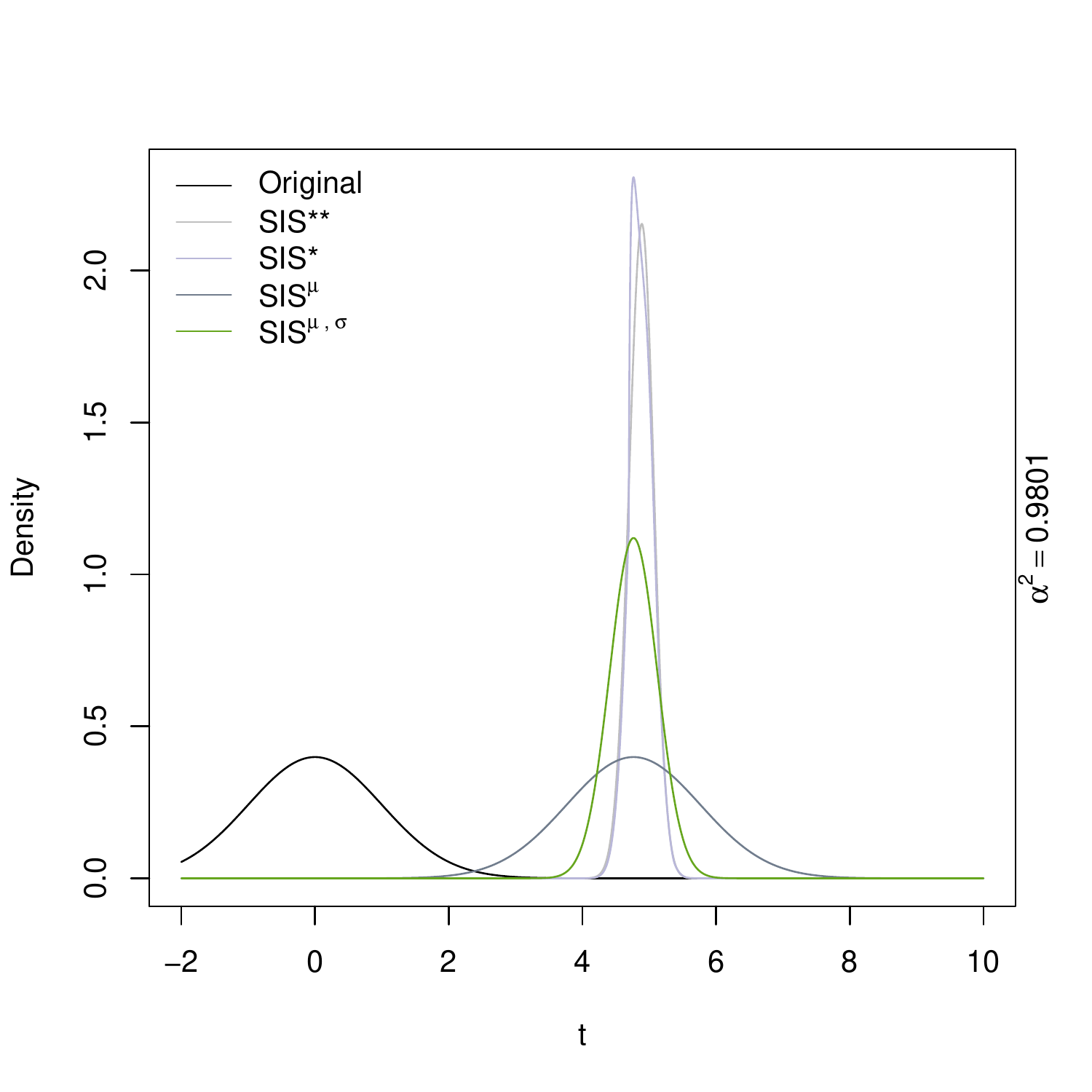}
\includegraphics[width=0.48\textwidth]{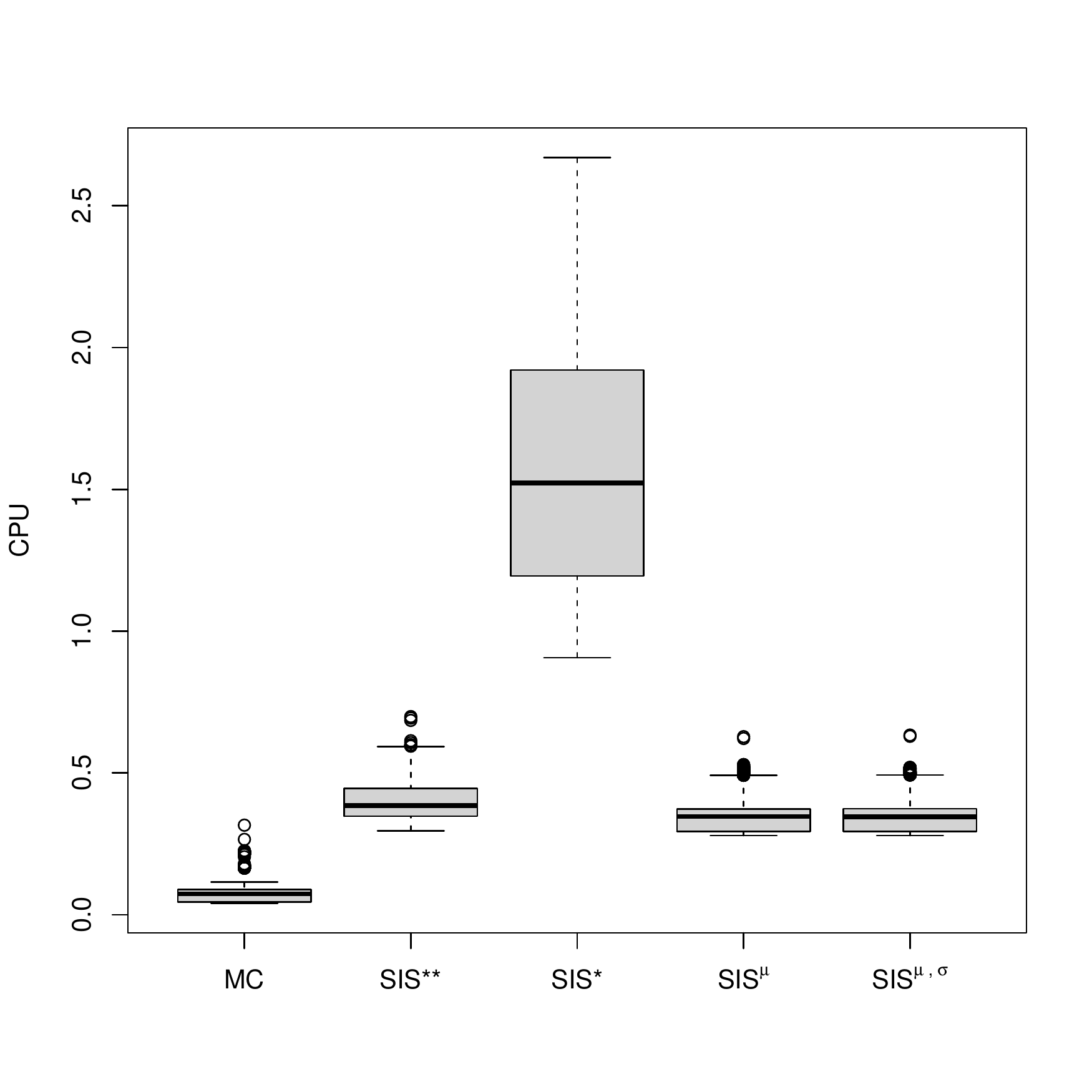}
\caption{Left: Calibrated densities for $\alpha=0.99$, $l=5$. Right: Run-times for each method including pilot runs.}
\label{fig:lm:densdur}
\end{figure}

Figure~\ref{fig:lm:results:pl} displays mean relative errors; recall that we know $p_l$ here. The relative errors for the different methods are similar, though $\SIS^{\mu,\sigma}$  seems to give smallest errors. A possible explanation might be that the simplicity of that method (e.g., in terms of the support) relative to numerically constructing the optimal density via NINGL might outweigh the benefit of the latter having slightly more theoretical support. Furthermore, note that the IS methods perform much better when $R^2=\alpha^2$ is larger, i.e., when the single index structure is strong, as expected.

\begin{figure}
\centering
\includegraphics[width=0.48\textwidth]{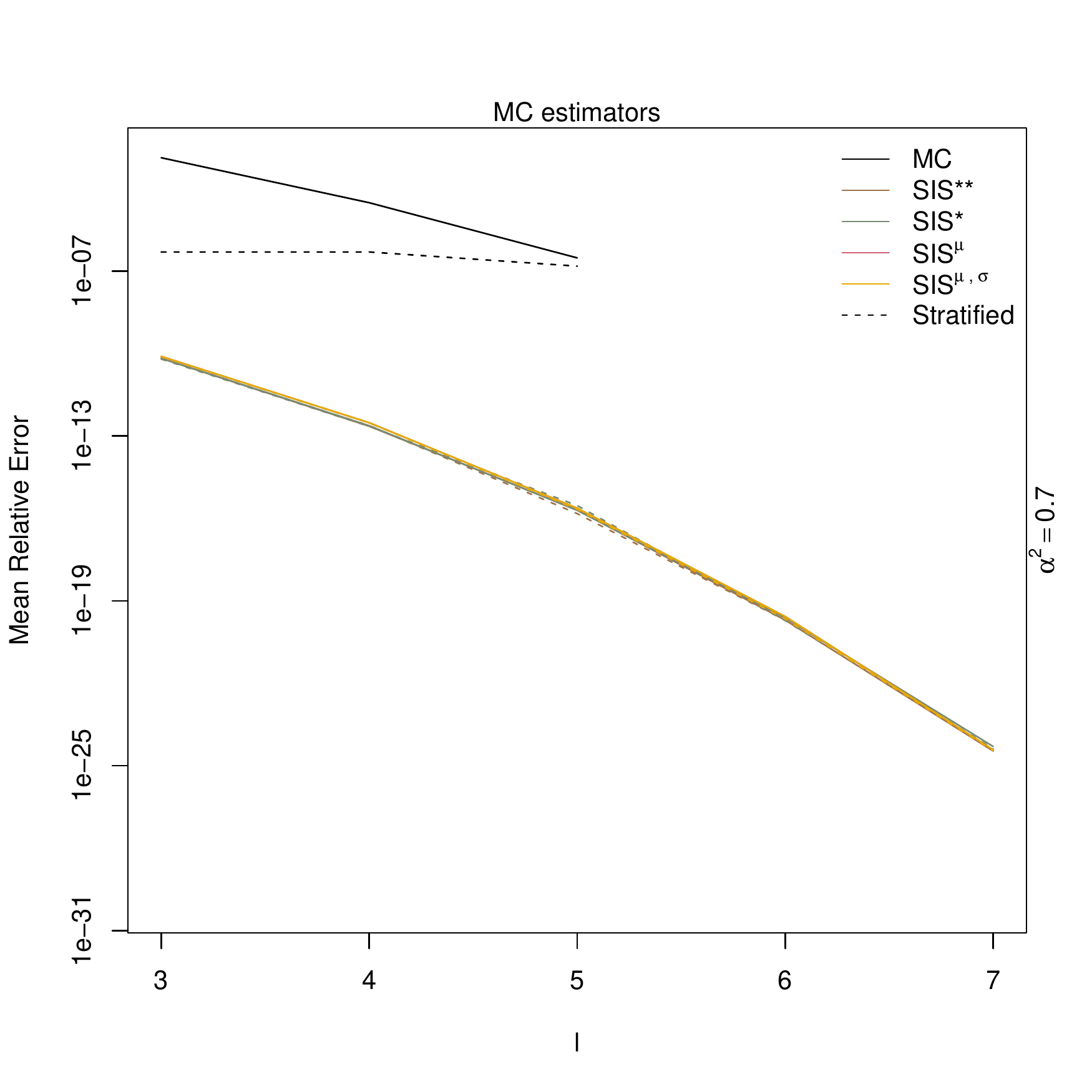}
\includegraphics[width=0.48\textwidth]{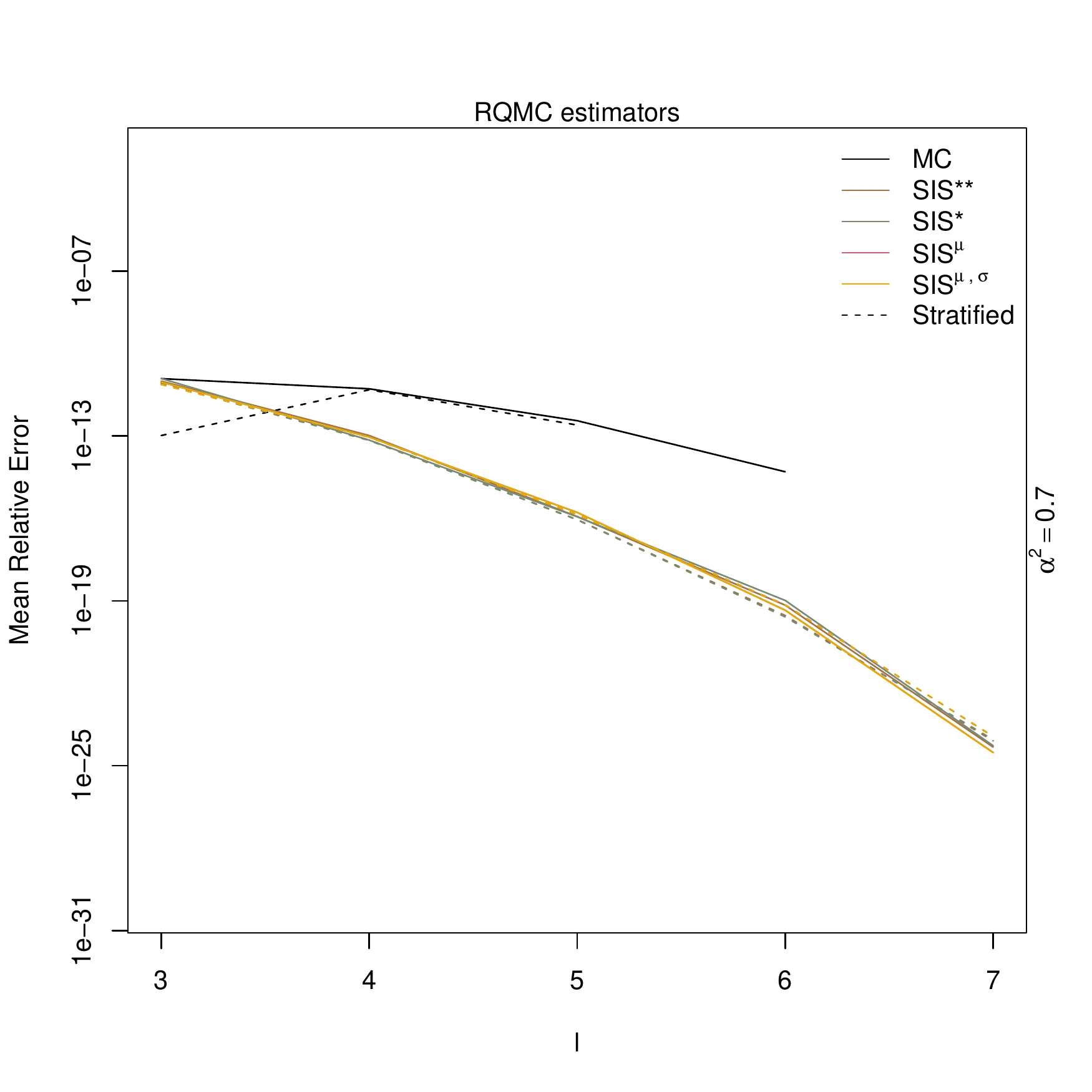}
\includegraphics[width=0.48\textwidth]{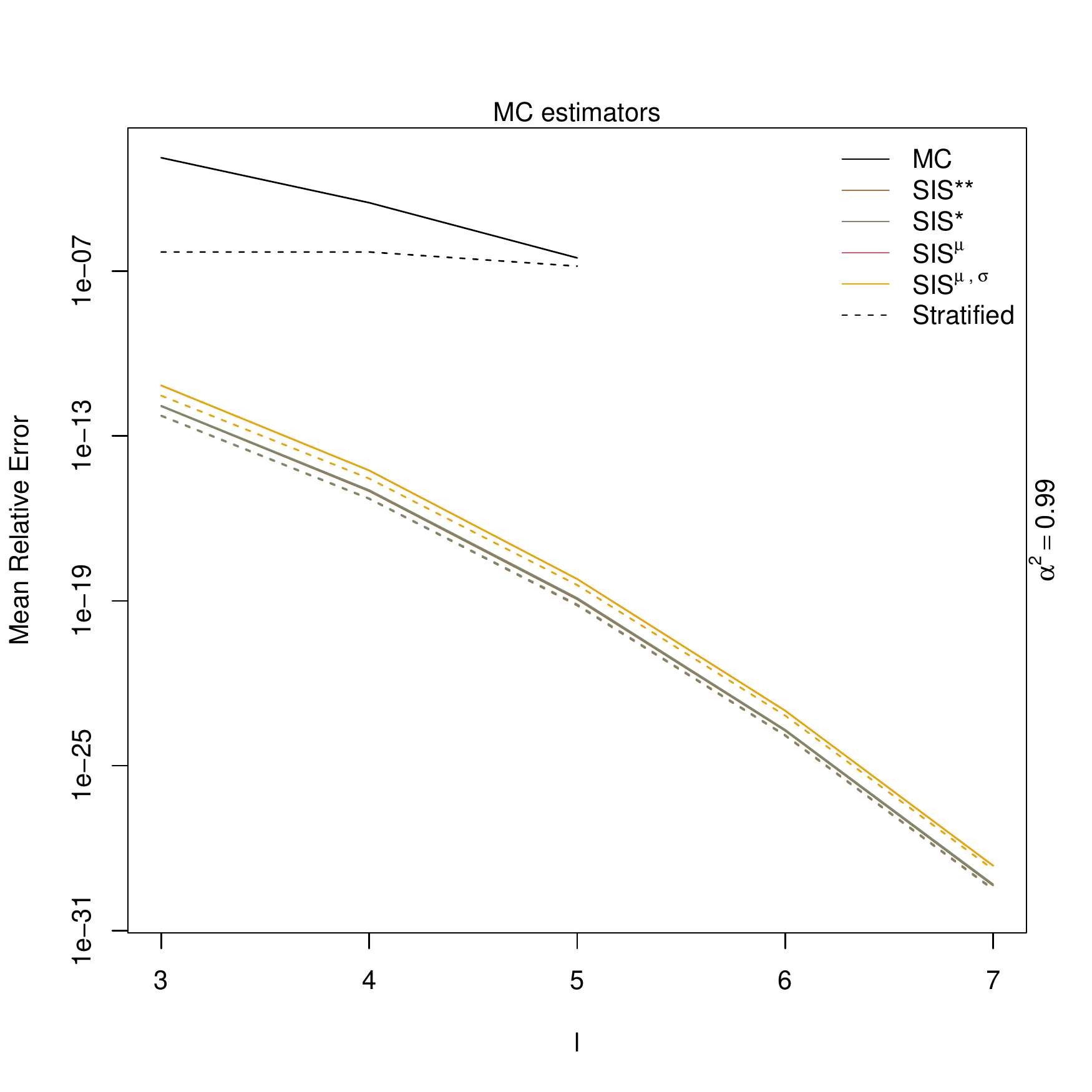}
\includegraphics[width=0.48\textwidth]{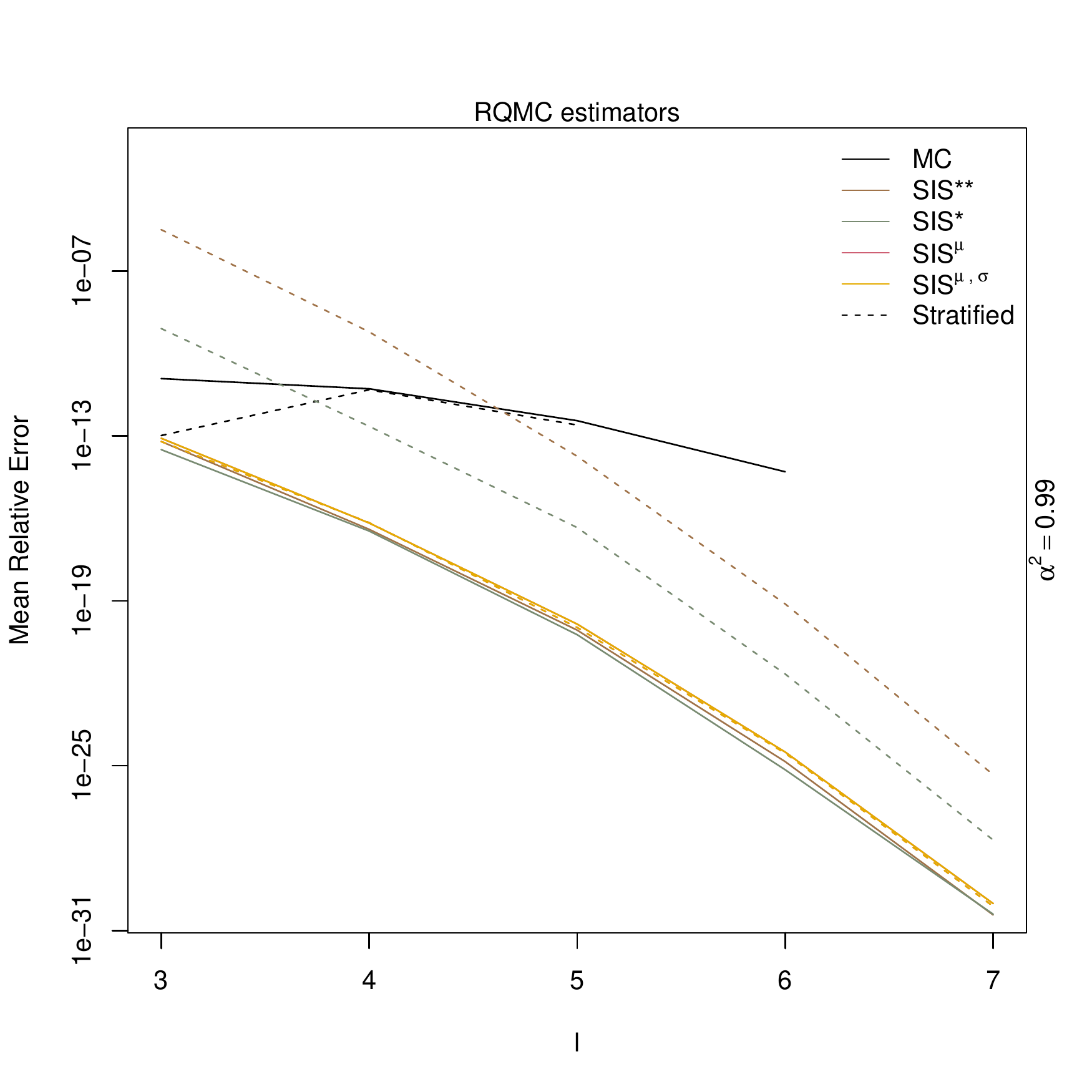}
\caption{Mean relative errors when using pseudo-random numbers (left) and quasi-random numbers (right) for
$\alpha^2=0.7$ (top) and $\alpha^2=0.99$ (bottom).}
\label{fig:lm:results:pl}
\end{figure}

\begin{figure}
\centering
\includegraphics[width=0.48\textwidth]{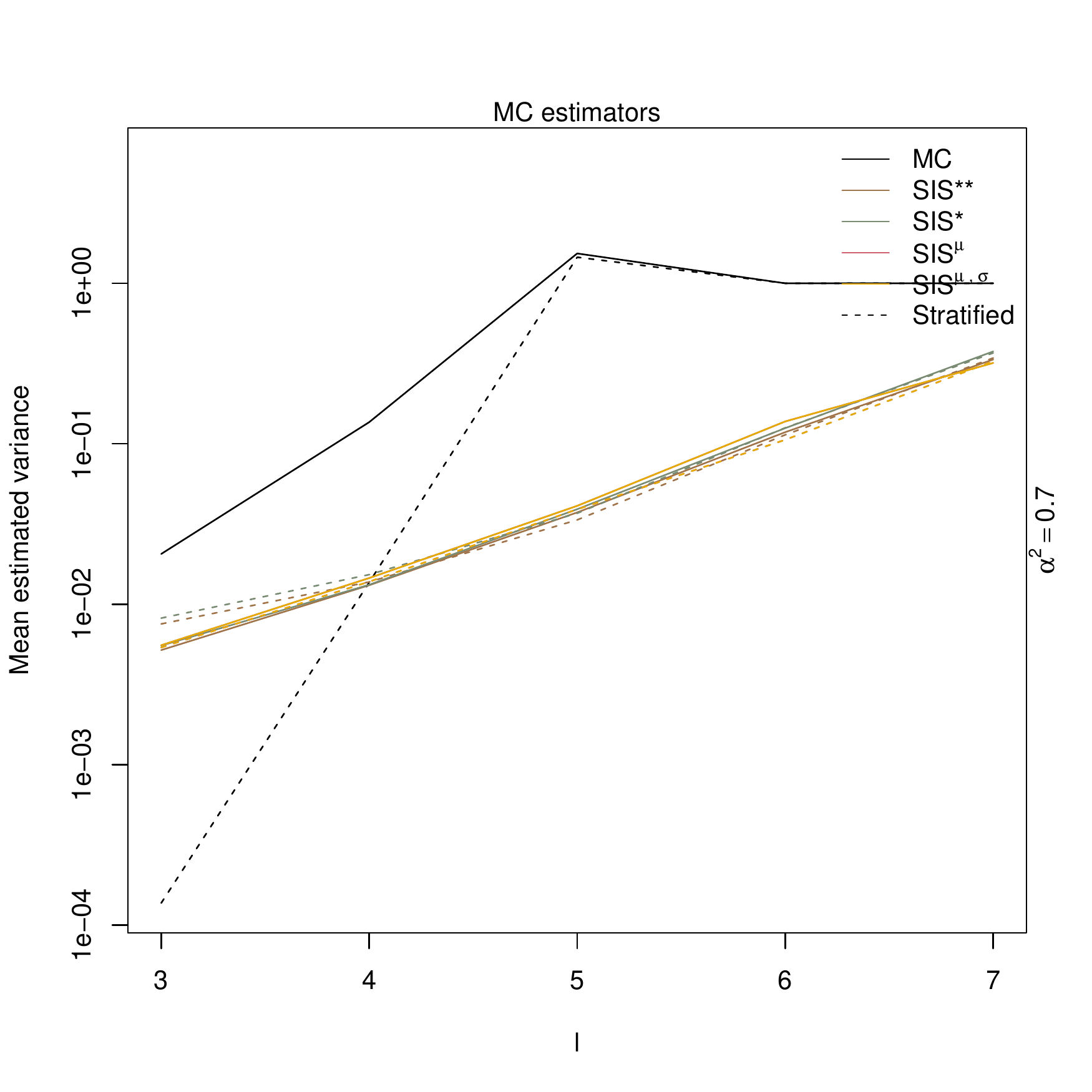}
\includegraphics[width=0.48\textwidth]{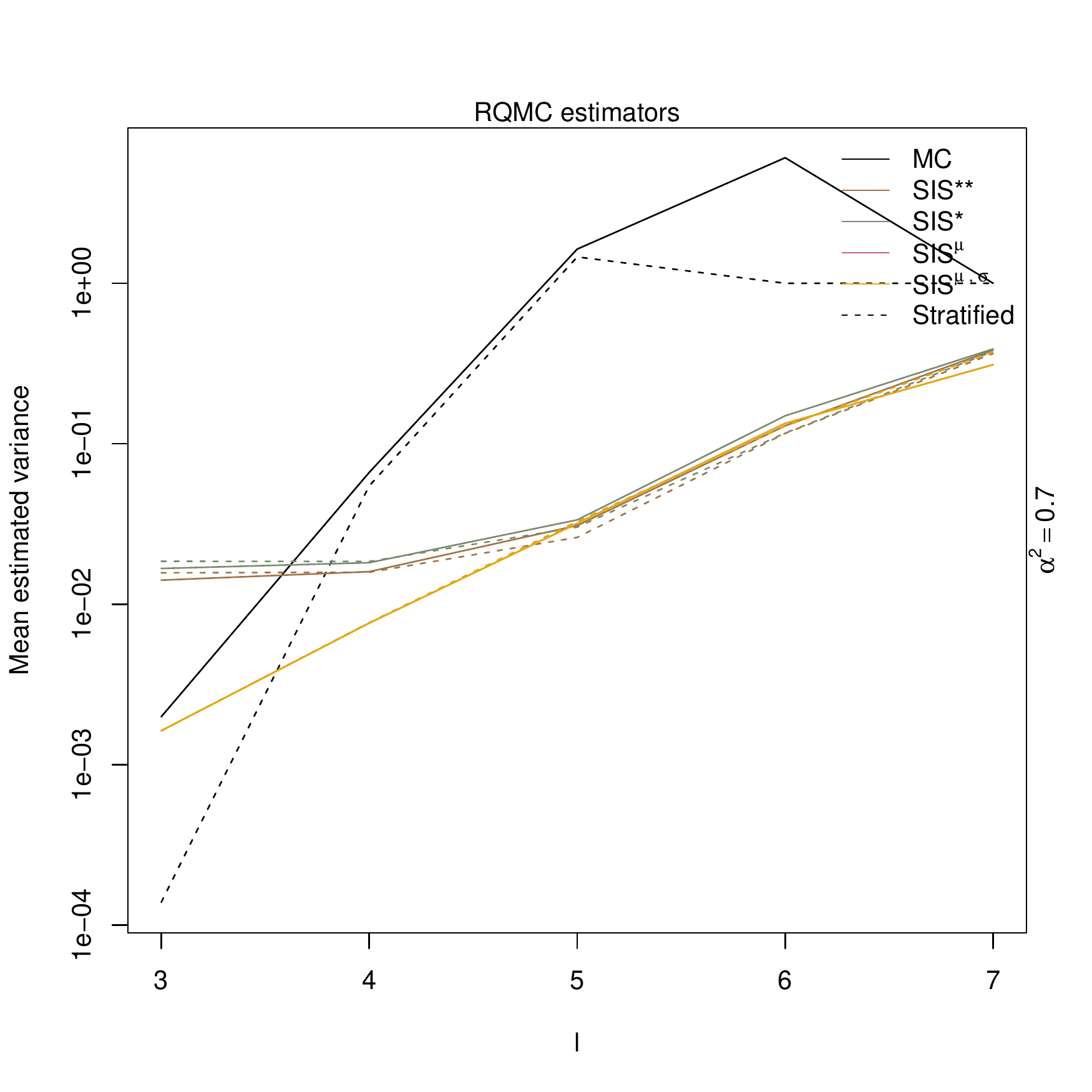}
\includegraphics[width=0.48\textwidth]{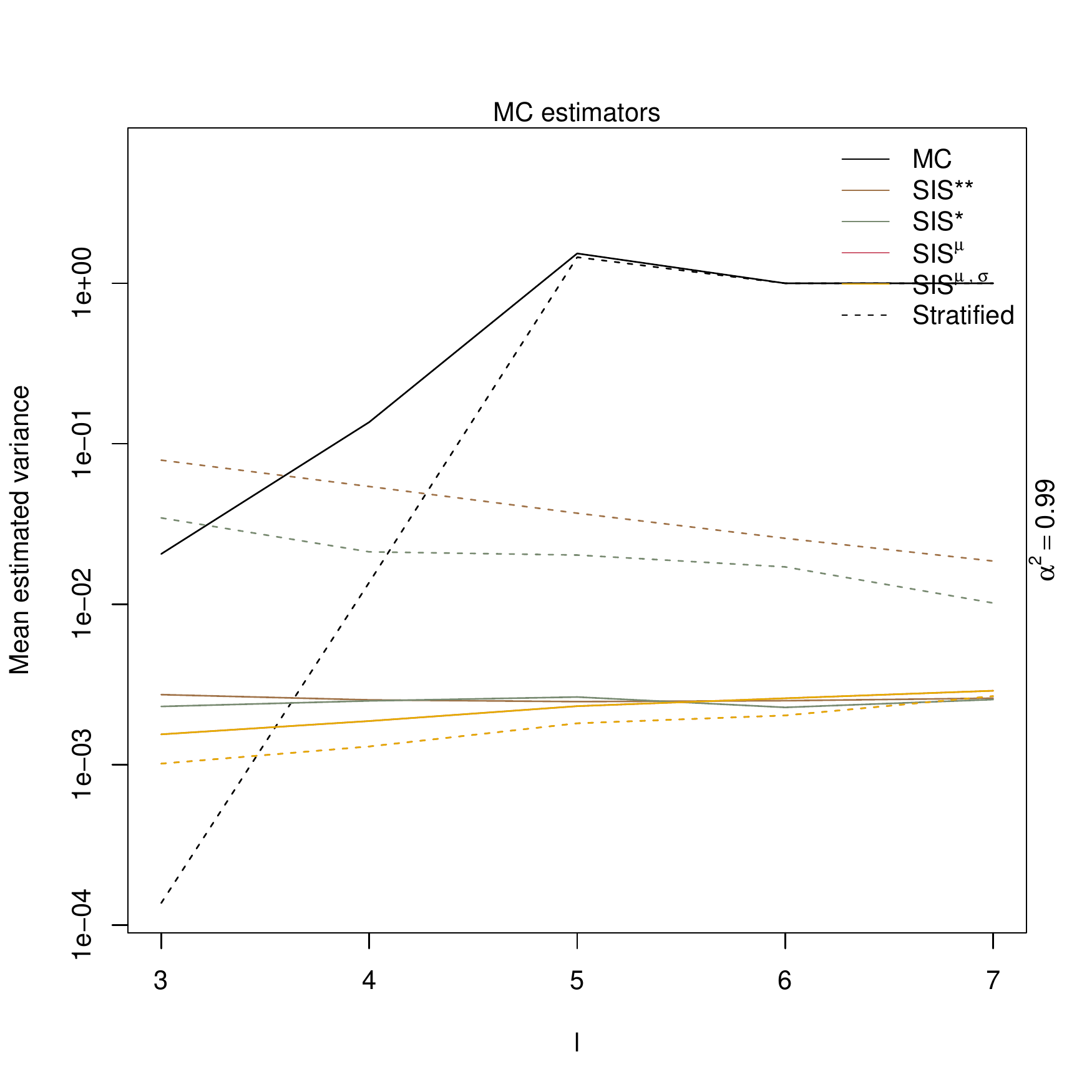}
\includegraphics[width=0.48\textwidth]{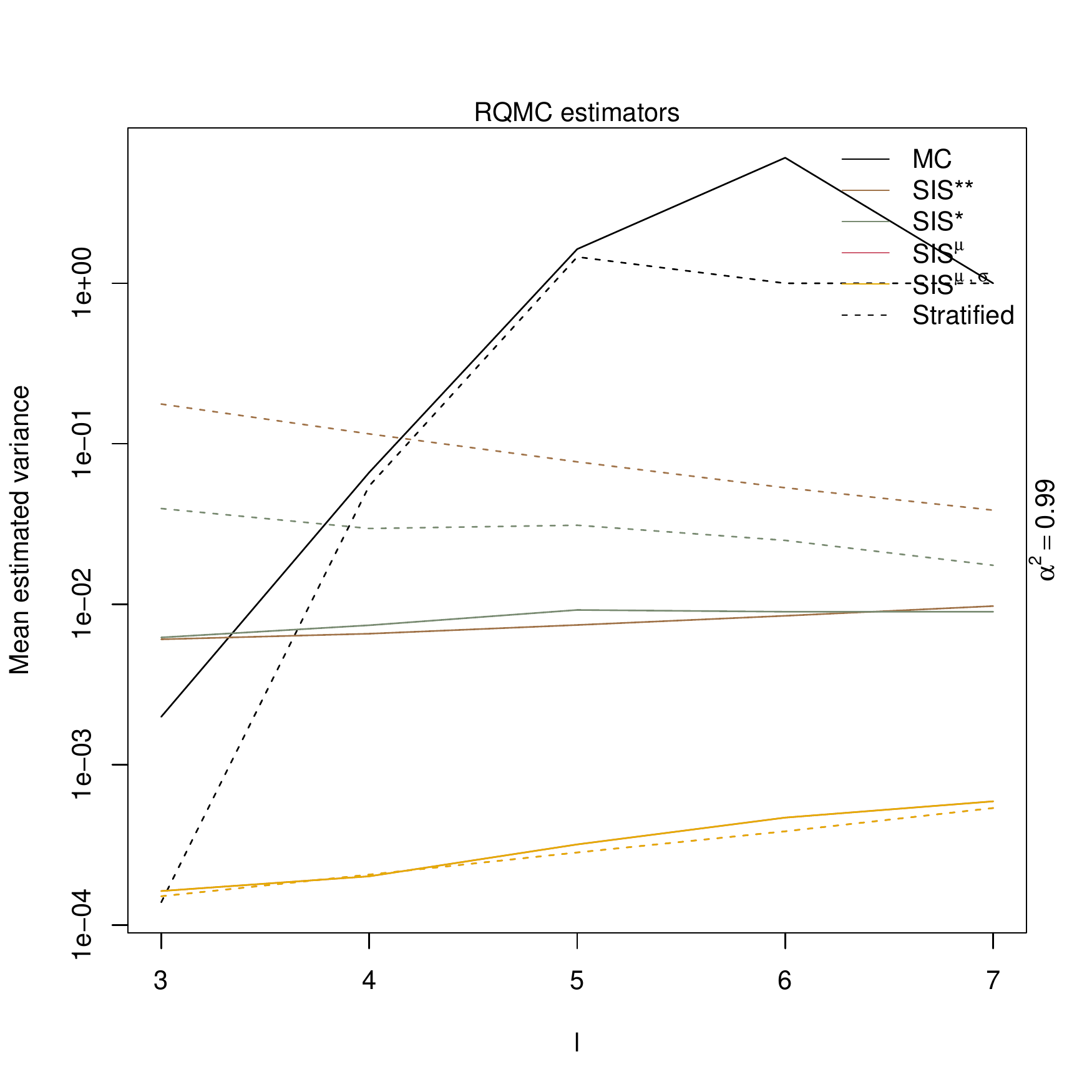}
\caption{Mean estimated variance when using pseudo-random numbers (left) and quasi-random numbers (right) for
$\alpha^2=0.7$ (top) and $\alpha^2=0.99$ (bottom).}
\label{fig:lm:results:pl}
\end{figure}

\subsection{Loss Distribution of a Credit Portfolio}\label{section:credit:gauss}
In this section, we study the effectiveness of the proposed methods for a credit portfolio problem studied in \cite{glassermanli2005}, where the goal is to estimate the probability of large portfolio losses under a normal copula model. We compare our proposed methods to the IS technique of Glasserman and Li, to which we refer to as G$\&$L IS.
\subsubsection{Problem Formulation}\label{section:credit:gauss:pars}
Suppose that $Y_k$ denotes the default indicator of the $k$th obligor with exposure $c_k$ and a default probability of $p_k$ for $k=1,\ldots,h$. The incurred loss is then $L = \sum_{k=1}^h c_k Y_k$. Let
$$Y_k = \I_{\{X_k > \Phi^{-1}(1-p_k)\}},\quad X_k = a_{k1} Z_1 + \cdots + a_{kd} Z_d + b_k \eps_k \sim \N(0,1),\quad k=1,\dots,h,$$
where
$$ (Z_1,\ldots,Z_d) \sim \N_d(\bm{0},I_d), \quad\eps_1,\dots,\eps_h\isim\N(0,1), \quad \sum_{j=1}^h a^2_{kj} \leq 1,\quad b_k = \sqrt{1- \sum_{j=1}^h a^2_{kj} }.$$
The $a_{kj}$ represent the $k$th obligor's factor loadings for the $d$ risky systematic factors; the choice of $b_k$ ensures $X_k\sim\N(0,1)$. Our goal is to estimate $P(L>l)$ for small $l>0$.

As in \cite{glassermanli2005},  we consider a portfolio with $h=1\ 000$ obligors in a 10-factor model (i.e. $d=10$). The marginal default probabilities and exposures are $p_k = 0.01 \cdot (1+\sin(16 \pi k /h))$ and $c_k = (\left \lceil{{5k/h}}\right \rceil)^2$ for $k=1,\ldots, h$, respectively. The marginal default probabilities vary between $0 \%$ and $2\%$ and the possible exposures are 1, 4, 9, 16 and 25, with 200 obligors at each level. The factor loadings $a_{kj}$'s are independently generated from a $\U(0,1/\sqrt{d})$.
Letting $\bZ=(Z_1,\ldots,Z_d)\T$ and $\bm{\eps}=(\eps_1,\ldots,\eps_h)\T$, we write $L=L(\bZ,\bm{\eps})$,
i.e.,  the vector $\bX$ to which we have referred throughout this paper is given by $\bX=(\bZ,\bm{\eps})$ for this example.
We investigate whether or not $L$ has a single-index structure. Let $T= \btheta\T \bZ$ where $\btheta \in \IR^d$ such that $\btheta\T \btheta=1$, so $T \sim \N(0,1)$. We estimate $\btheta$ that maximize the fit by using the average derivative method of \cite{stoker1986}. The estimated $\btheta$ has almost equal entries close to $\sqrt{1/d}$. This makes intuitive sense, as each component of $\bZ$ is likely to be equally important because the factor loadings are generated randomly.
The left side of Figure~\ref{fig:scatter:optdens:gcop} shows the scatter plot of $(T,L)$.
The figure reveals the single-index model fits $L$ well even in the extreme tail, implying SIS based on this choice of $T$ will give substantial variance reduction.
The right side of Figure~\ref{fig:scatter:optdens:gcop} displays the original density of $T$, the optimally calibrated $\SIS^*$ density as well as the estimated function $p_l(t)$. Note that the optimally calibrated density's mode substantially differs from the original one.

\begin{figure}[h]
	\centering
	\includegraphics[width = 0.48\textwidth]{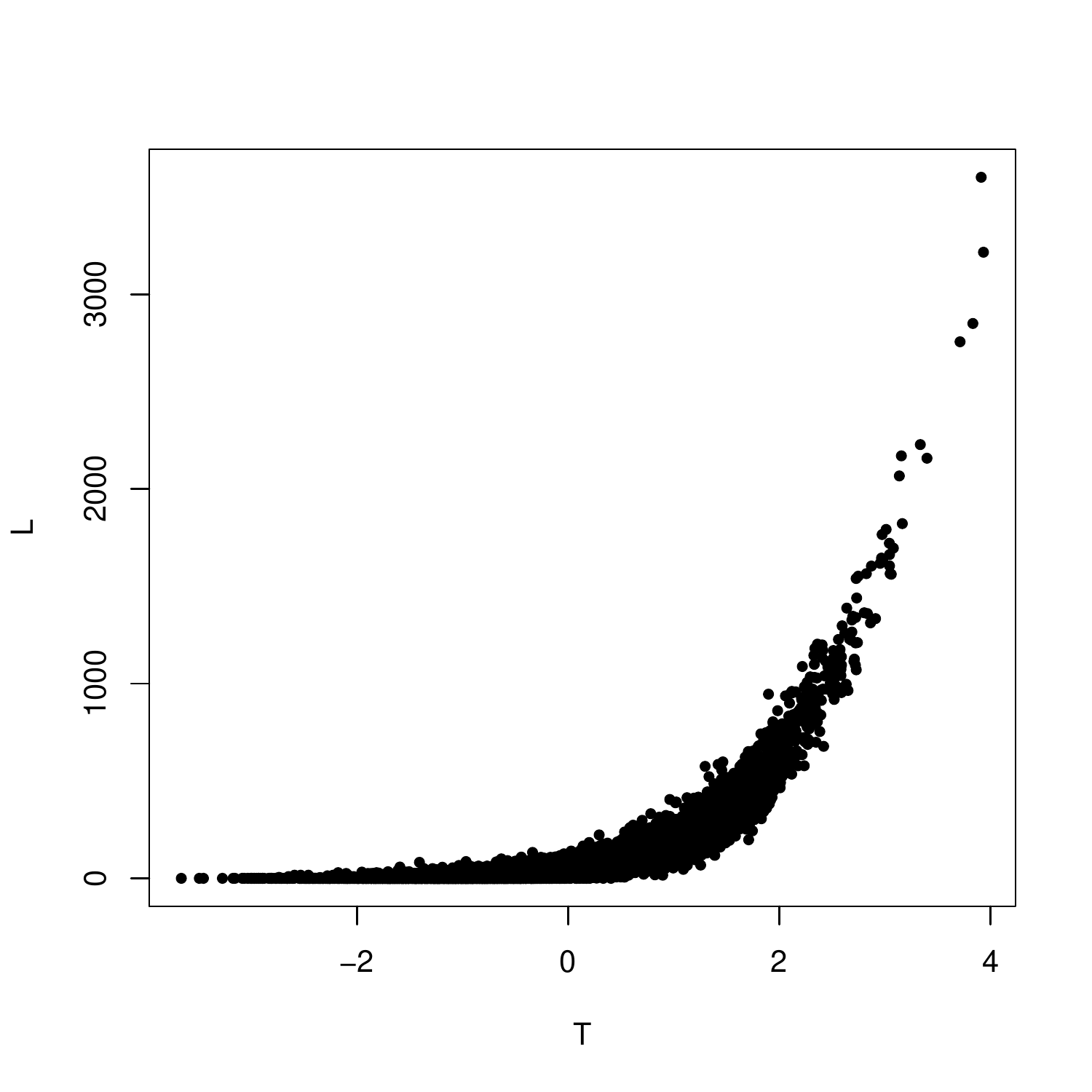}
	\includegraphics[width = 0.48\textwidth]{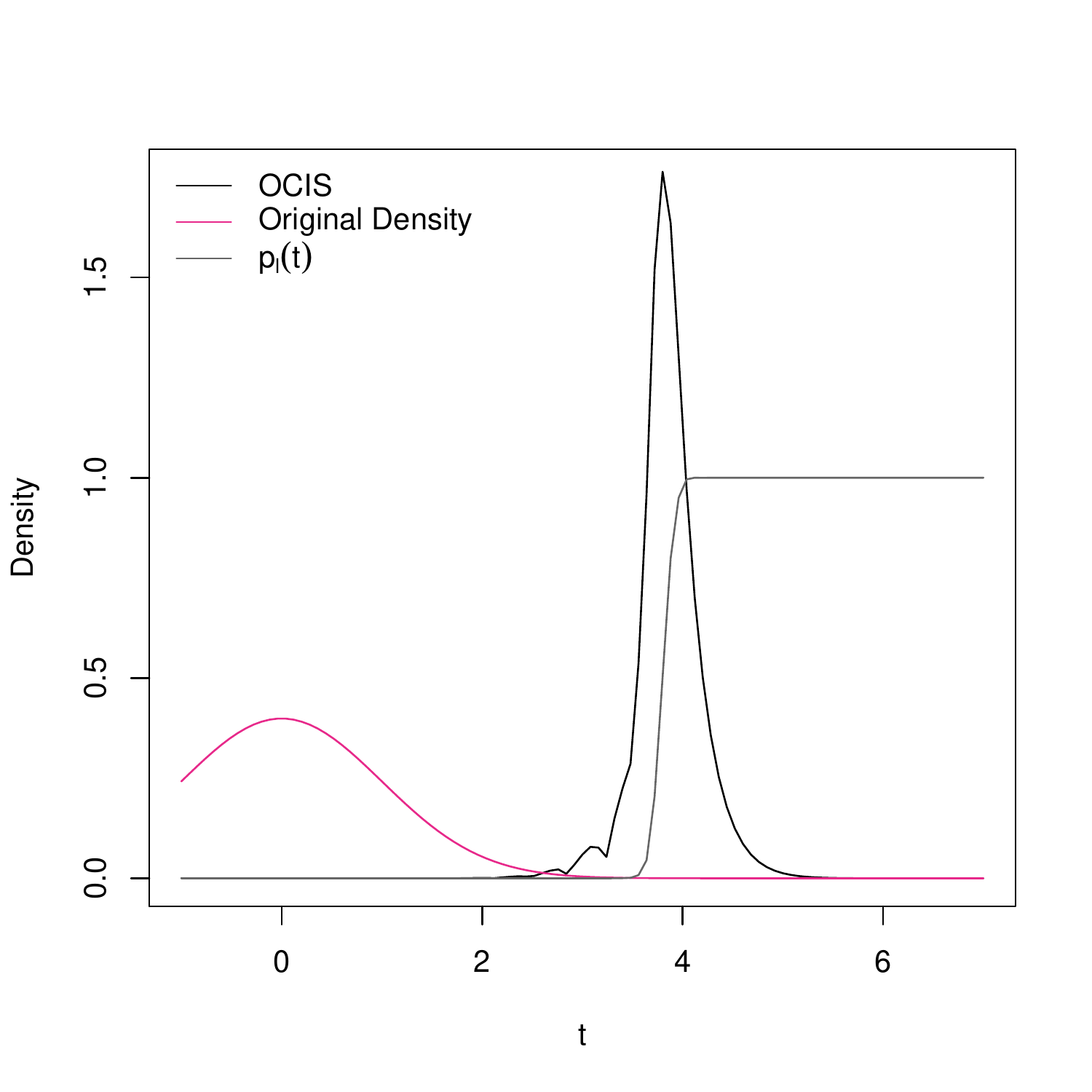}
	\caption{Plot of Transformed variable ($T$) vs Portfolio Loss ($L$) based on 10\ 000 observations (left) and OCIS density calibrated to $l=3000$ (right).}
	\label{fig:scatter:optdens:gcop}
\end{figure}

\subsubsection{Proposed estimators}
The method of \cite{glassermanli2005} consists of a two-step procedure. In a calibration stage, an optimal mean vector $\bmu\in\IR^d$ is found by solving an optimization problem minimizing the variance of the resulting IS estimator. Next, one samples $\bZ\sim\N_d(\bmu, I_d)$ and computes the conditional default probabilities $p_k(\bZ)=\P(Y_k=1\mid \bZ)=\Phi( (\ba_k\T\bZ - x_k)/b_k)$, which enter another optimization problem used to find a number $\theta\in\IR$ so that $q_k(\theta, \bZ)$ are variance minimizing default probabilities. Given $\bZ$, we know that $Y_1,\dots,Y_h$ are independent and can therefore easily sample the loss via $L=\sum_{k=1}^h c_k \I_{\{U_k \leq q_k(\theta(\bZ)\}}$ where $(U_1,\dots,U_h)\sim \U(0,1)^h$. Finally, the estimator $\I_{\{L>l\}}\cdot w(\bZ, L)$ where $w$ denotes the IS weight function is an unbiased estimator.

Our method $\SIS^{\mu,\sigma}$ proceeds as described in \ref{alg:SIS:full}; $\SIS^{\mu}$ omits sets the scale to unity while the $\SSIS$ methods also stratify. Once $\bZ\mid T$ is sampled, we sample $Y_k$ from $p_k(Z)$ independently. We also include $\SIS^*$ and $\SSIS^*$, where the function $p_l(t)$ is estimated as before and the quantile function of the optimal distribution is estimated via the NINiGL algorithm, in our experiments; see also Figure~\ref{fig:scatter:optdens:gcop}.

\subsubsection{Comparison}
We compare SIS and SSIS to G\&L IS by computing estimates, standard errors and computation times for $l \in \{100, 1000, 2000, 3000, 4000\}$. All methods require a calibration stage. For this comparison, we optimize the proposal distributions at each loss level of $l$ separately and estimate the corresponding loss probability.
Table~\ref{table:res:gausscredit} shows the estimated probabilities along with half-widths of estimated confidence intervals (CI) in brackets. The last column shows the average computational time of each method over all loss levels $l$. All examples used $n=5000$ samples and $1000$ samples for the calibration.

We see that all our methods lead standard errors smaller than G\&L IS, while the estimated CIs for both methods are typically overlapping, supporting the correctness of both approaches. Given the small run-time, unbiasedness and small estimated errors, we can conclude that $\SSIS^{\mu,\sigma}$ is the best estimator for this problem. This supports our claim that the optimal density of $T$ can be quickly and accurately approximated by a location scale transform of $f_T$. Note that $\SIS^*$ and $\SSIS^*$ are particularly slow, as it involves numerically approximation the quantile function corresponding to the optimal $g_T$.

\begin{table}[ht]
	\centering
	\begin{tabular}{c|cccccc}
	  	\toprule
	l & 100 & 1000 & 2000   &3000 &4000   & Avg run-time (sec)\\
	\midrule
G\&L IS     & 0.28  &0.0079 &0.00077&9.2e-05&1.1e-05&2.45\\
	& (0.0078) & (0.00036) & (4.1e-05) & (6.3e-06) & (8.8e-07) & \\
$\SIS^*$      &0.28&  0.0081&0.00076&9.2e-05&1.1e-05&6.62\\
		& (0.0068)& (0.00021)& (2.1e-05) & (2.4e-06) & (3.5e-07)& \\
$\SSIS^*$       &0.28&  0.0082&0.00077&9.5e-05&1.1e-05&12.56\\
			& (0.0046) & (0.00014) & (1.4e-05) & (1.7e-06) & (2.5e-07) & \\
$\SIS^{\mu}$      &0.28&  0.0077&0.00074&8.6e-05&1e-05   &1.41\\
				& (0.0086) & (0.00039) & (4.2e-05) &(5.5e-06) & (6.8e-07) & \\
$\SSIS^{\mu}$   &0.28&  0.008& 0.00075&9.1e-05&1.1e-05&1.45\\
					& (0.0062) & (0.00028) & (2.9e-05) &(4e-06)&(5.1e-07) & \\
$\SIS^{\mu,\sigma}$ &0.28&  0.0082&0.00077&9.4e-05&1.1e-05&2.45\\
						&(0.0077)&(0.00034) & (3.3e-05) & (5.2e-06) &(4.6e-07) & \\
$\SSIS^{\mu,\sigma}$ &0.28&  0.0081&  0.00075&8.9e-05&1.1e-05&   2.2\\
							&(0.0059) &(2e-04) & (1.9e-05)& (2.3e-06) & (3e-07) & \\
	\bottomrule
	\end{tabular}
		\caption{Estimates and CI halfwidths when estimating $p_l$ in the Gaussian Credit Portfolio problem with $h=1000$ obligors and $d=10$ factors for various $l$
	and methods. The last column displays average run-times. }
	\label{table:res:gausscredit}
	\end{table}

	          \begin{table}[ht]
	\centering
	\begin{tabular}{c|cccccc}
	  	\toprule
	l & 100 & 1000 & 2000   &3000 &4000   \\
	\midrule
G\&L IS     & 1.5 & 1.5 & 1.5 & 1.5 & 1.6 \\
$\SIS^*$  & 1.3 & 1.3 & 1.2 & 1.7 & 1.5 \\
$\SIS^{\mu,\sigma}$ & 1.6 & 1.5 & 1.9 & 1.7 & 1.6 \\
$\SSIS^{\mu,\sigma}$ & 1 & 1.1 & 1 & 1.1 & 0.9 \\
	\bottomrule
	\end{tabular}
		\caption{Relative error reduction factors RE(MC)/RE(RQMC) for the Gaussian credit portfolio with
		$h=1000$ obligors and $d=10$ factors for various $l$
	and methods. }
	\label{table:vrf:gausscredit}
	\end{table}

\subsection{Tail probabilities of a $t$-Copula Credit Portfolio}\label{section:credit:t}
In this section, we apply SIS to a credit portfolio problem under a $t$-copula model, which is the model studied in Section~\ref{section:credit:gauss} with a multiplicative shock variable included. This $t$-copula model is a special case of the models with extremal dependence studied in \cite{bassamboojunejazeevi2008}. Unlike the Gaussian copula, the $t$-copula is able to model tail dependence of latent variables, so simultaneous defaults of many obligors are more likely under the $t$-copula model than under its Gaussian copula counterpart.

\subsubsection{Problem Formulation}
In the $t$-copula model, the latent variables $\bX=(X_1,\ldots,X_d)$ are multivariate-$t$ distributed, that is,
\begin{align*}
X_k = \sqrt{W} (a_{k1} Z_1 + \cdots + a_{kd} Z_d + b_k \eps_k),\quad k=1,\dots,h,
\end{align*}
where $W \sim\IG(\nu/2,\nu/2)$ is independent of $Z_1,\dots,Z_d, \eps_k\isim\N(0,1)$. Accordingly, we define $Y_k = \I_{\{X_k > t^{-1}_{\nu}(1-p_k)\}}$. We assume the same parameters as in Section~\eqref{section:credit:gauss:pars}, except that now we have $h=50$ obligors, and the two settings for the degrees-of freedom $\nu\in\{5, 12\}$. Let $\bZ=(Z_1,\ldots,Z_d)\T$ and $\bm{\eps}=(\eps_1,\ldots,\eps_h)$. We consider two transformations.
For the first transformation,  let $Z_W= \Phi^{-1}(F_W(W))$ and
$$T_1(W, \bZ, \bm{\epsilon})=\beta_W Z_W + \bbeta_{L}\T \bZ,$$
where $\beta_W\in\IR$ and $\bbeta_{L}\in\IR^d$ are such that $\beta^1_W+\bbeta_{L}\T\bbeta_{L}=1$. Then, $T_1 \sim \N(0,1)$ since $Z_W\sim\N(0,1)$ is independent of $\bZ$.

Our second transformation relies on the random variable $S_l(\bZ,\bm{\eps})=\P(L>l \mid \bZ,\bm{\eps})$ and note that $\P(L>l)=\E(S_l(\bZ,\bm{\eps}))$. Based on this and the fact that, given a sample $\bZ, \bm{eps}$, the function $S_l$ can be computed analytically, \cite{chankroese2010} propose to use CMC, i.e., estimating $\P(L>l)$ by the sample mean of $S_l(\bZ_i,\bm{\eps}_i)$ for independent $\bZ_i,\bm{\eps}_i$ for $i=1,\dots,n$. We propose to use this CMC idea combined with SIS by using the transformation
$$T_2=\bbeta_S\T\bZ$$
with $\bbeta_S$ such that $\bbeta_S\T\bbeta_S=1$, which implies $T_2\sim\N(0,1)$.

The second method based on CMC, is very effective as the variable $W$ which accounts for a large portion of the variance of $L$, is integrated out. Furthermore, \cite{chankroese2010} additionally employ IS on $(\bZ,\bm{\eps})$ to make the event $\{L>l\}$ more frequent using the cross-entropy method; see \cite{deboerkroesemannorrubinstein2005, rubinstein1997,rubinsteinkroese2013}. We refer to Chan and Kroese's method as C\&K CMC+IS. The numerical study in~\cite{chankroese2010} demonstrates that C\&K CMC+IS achieves substantial variance reduction. We will show in our numerical examples below that combining their CMC idea with our proposed single index IS method gives even greater variance reduction.

\subsubsection{Fit of Single-Index models with and without conditional Monte Carlo}
We first investigate whether or not $L$ and $S_l$ have single-index structures.
As before, the coefficients $\bbeta$ that maximize the fit of the single-index model are estimated using the average derivative method of Stoker \cite{stoker1986}.

\begin{figure}[htbp]
	\centering
	\includegraphics[width = 0.48\textwidth]{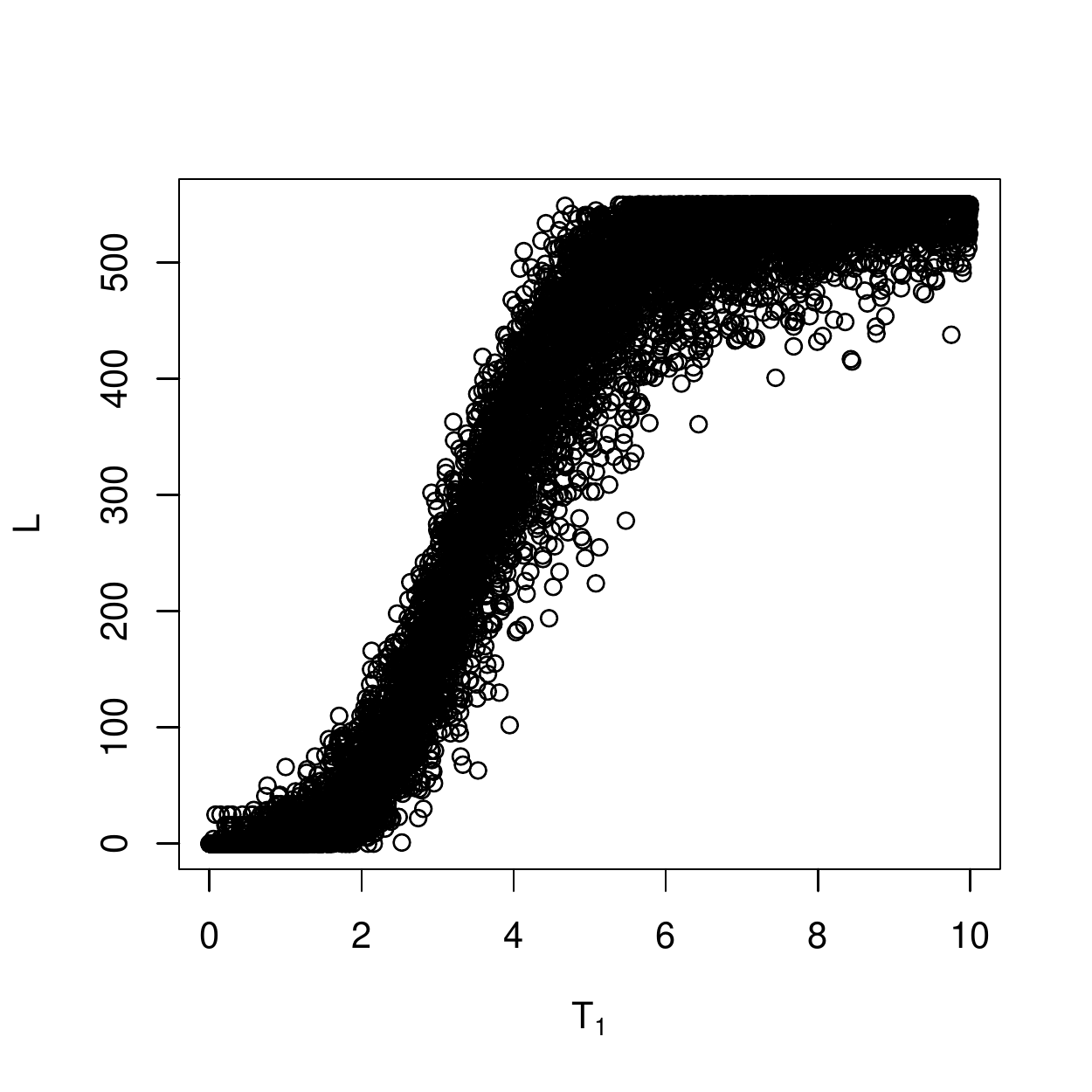}
	\includegraphics[width = 0.48\textwidth]{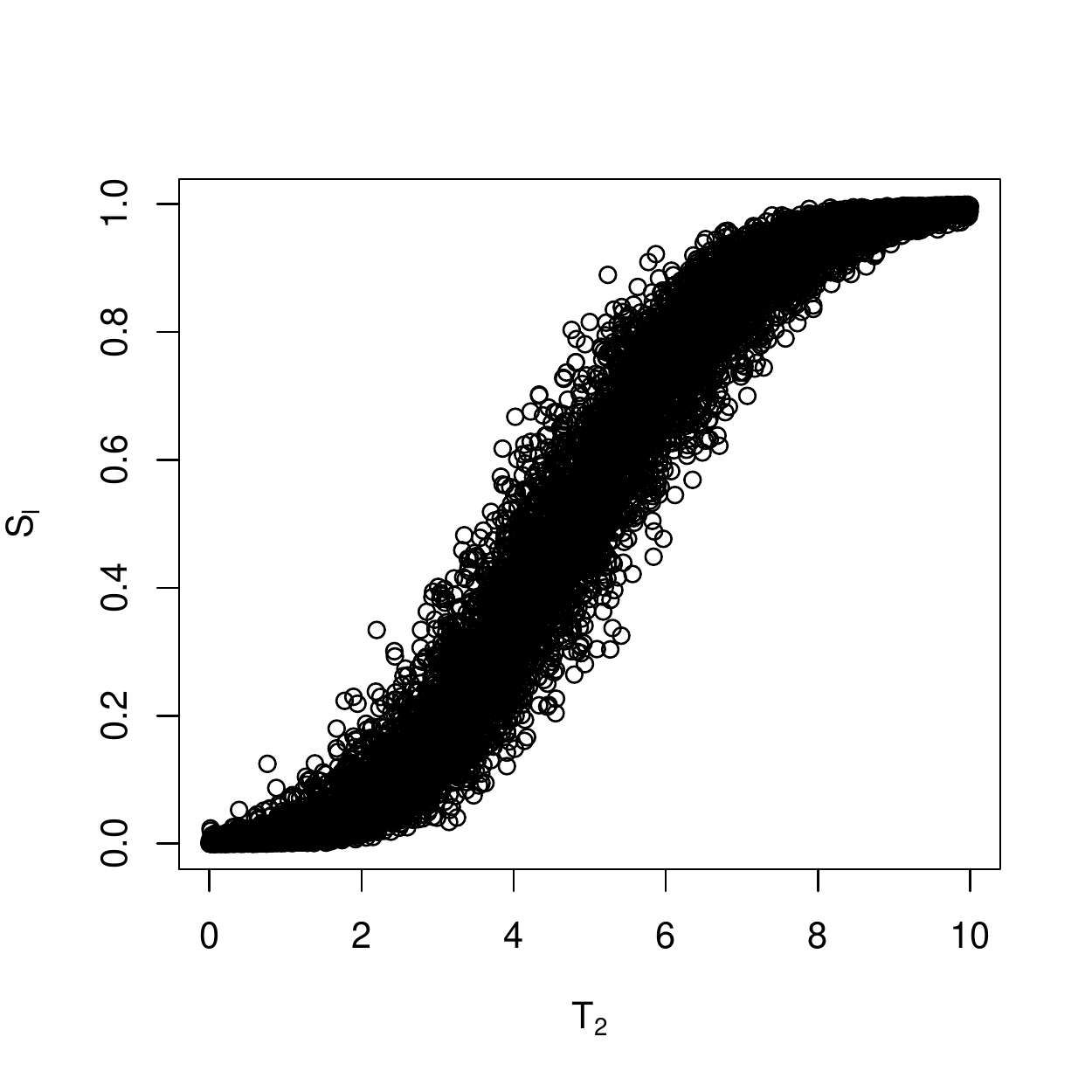}
	\includegraphics[width = 0.48\textwidth]{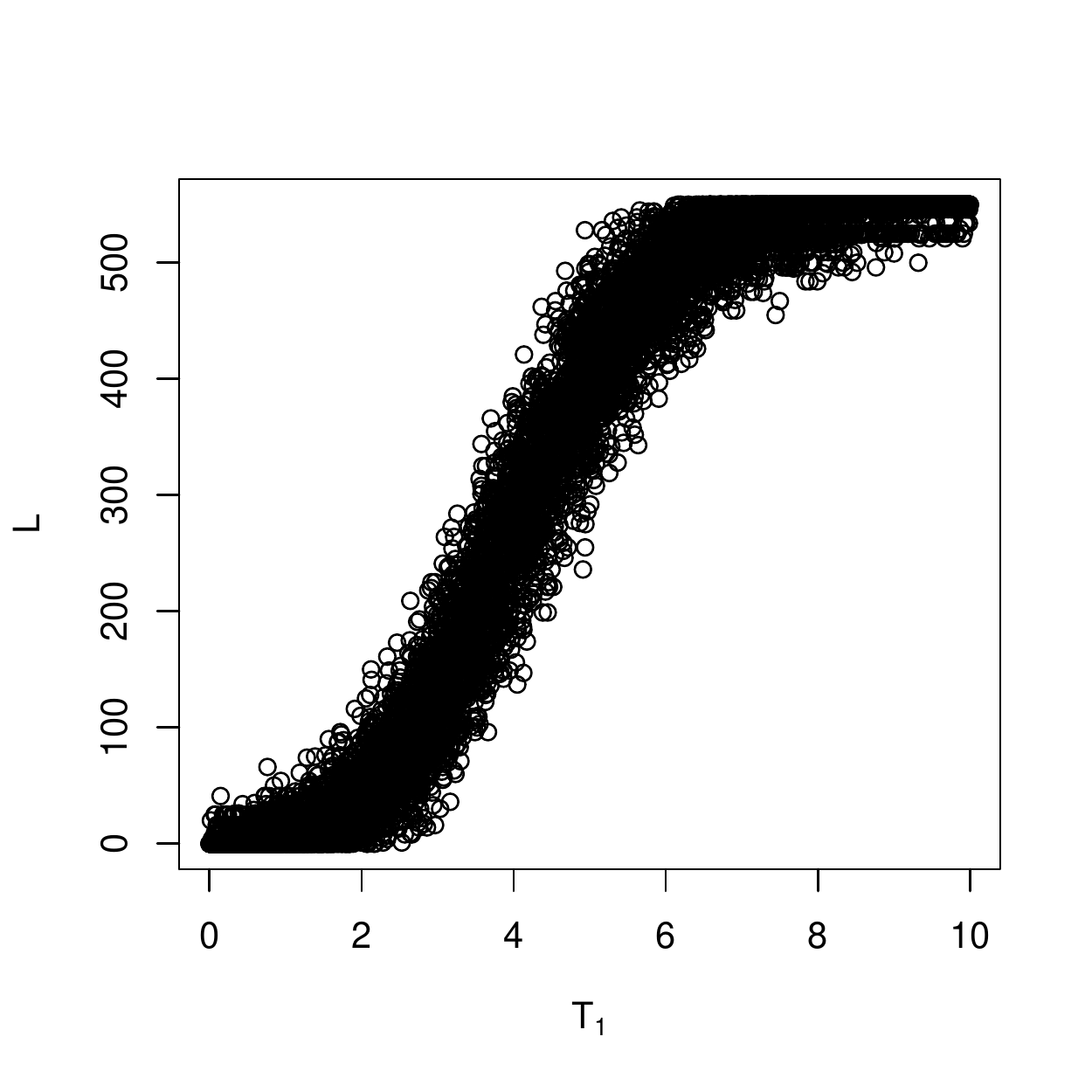}
	\includegraphics[width = 0.48\textwidth]{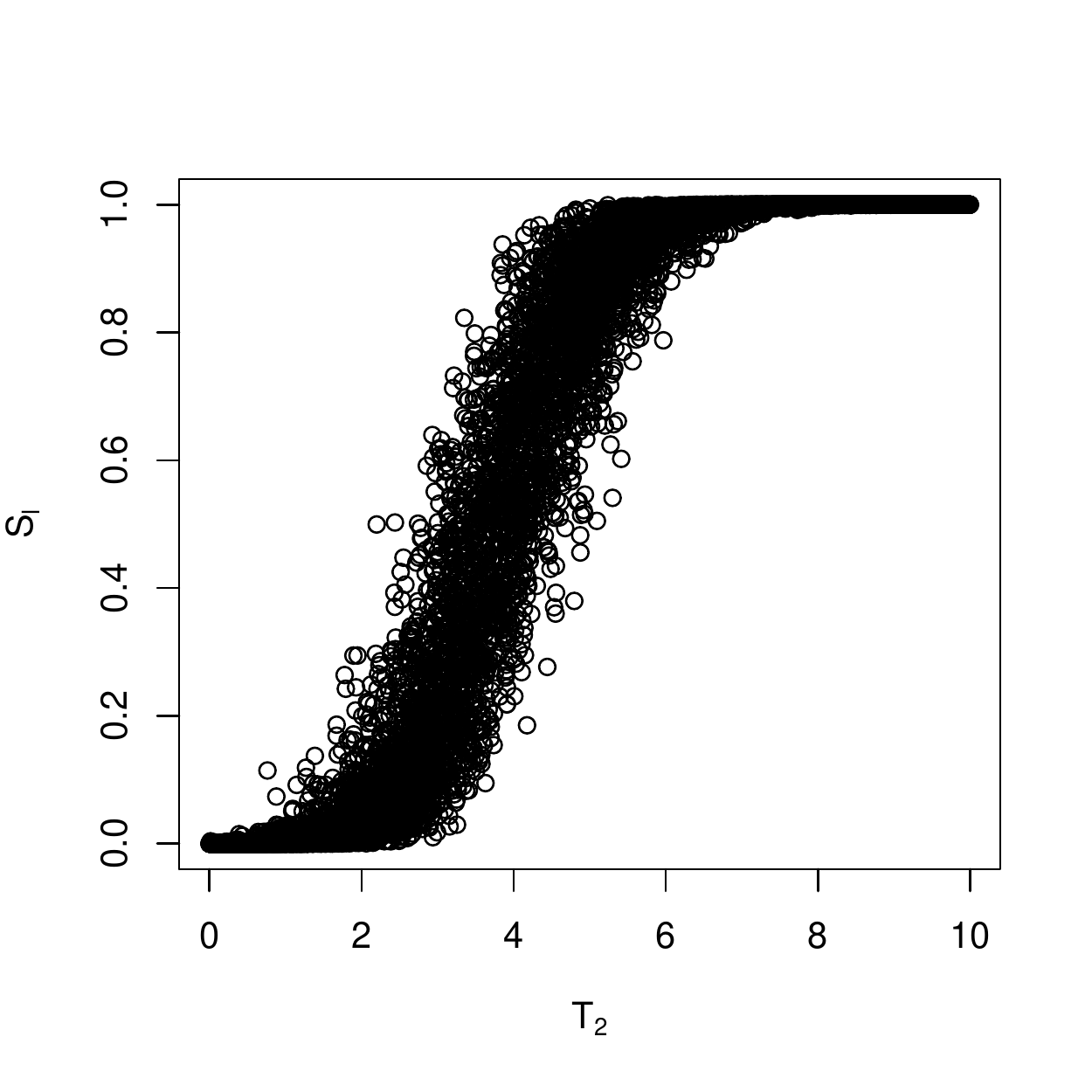}
		\caption{Scatter plots of $L$ vs $T_1$ (left) and $S_l$ vs $T_2$ (right) where $l=500$ and $\nu = 5$ (top) and $\nu = 12$ (bottom). }
	\label{fig:T_vs_L_tcredit}
\end{figure}

Figure~\ref{fig:T_vs_L_tcredit} shows scatter plots of $(T_1,L)$ and $(T_2, S_l)$ for $\nu=12$ and $\nu=5$.
The figures show that there is a strong association between $T_1$ and $L$ but the dependence is stronger when $\nu=12$ than when $\nu=4$. When $\nu=4$, there is a significant variation of $L$ that cannot be captured by the single-index model based on $T_1$ in the right-tail. This observation holds more generally; the smaller $\nu$ (i.e., the stronger the dependence between the $\bX_i$), the worse the fit of the single-index model becomes in the right-tail.
When investigating the fit of $(T_2, S_l)$, recall that the main advantage of CMC is that $W$ is integrated out; the resulting estimators should be less sensitive to the degrees-of-freedom $\nu$, which is the case in the plot. We can see that the fit of $T_2$ is excellent even in the outer right-tail for all settings of $\nu$ and $l$.

\subsubsection{Estimates and estimated variances}

\begin{figure}
\includegraphics[width = 0.48\textwidth]{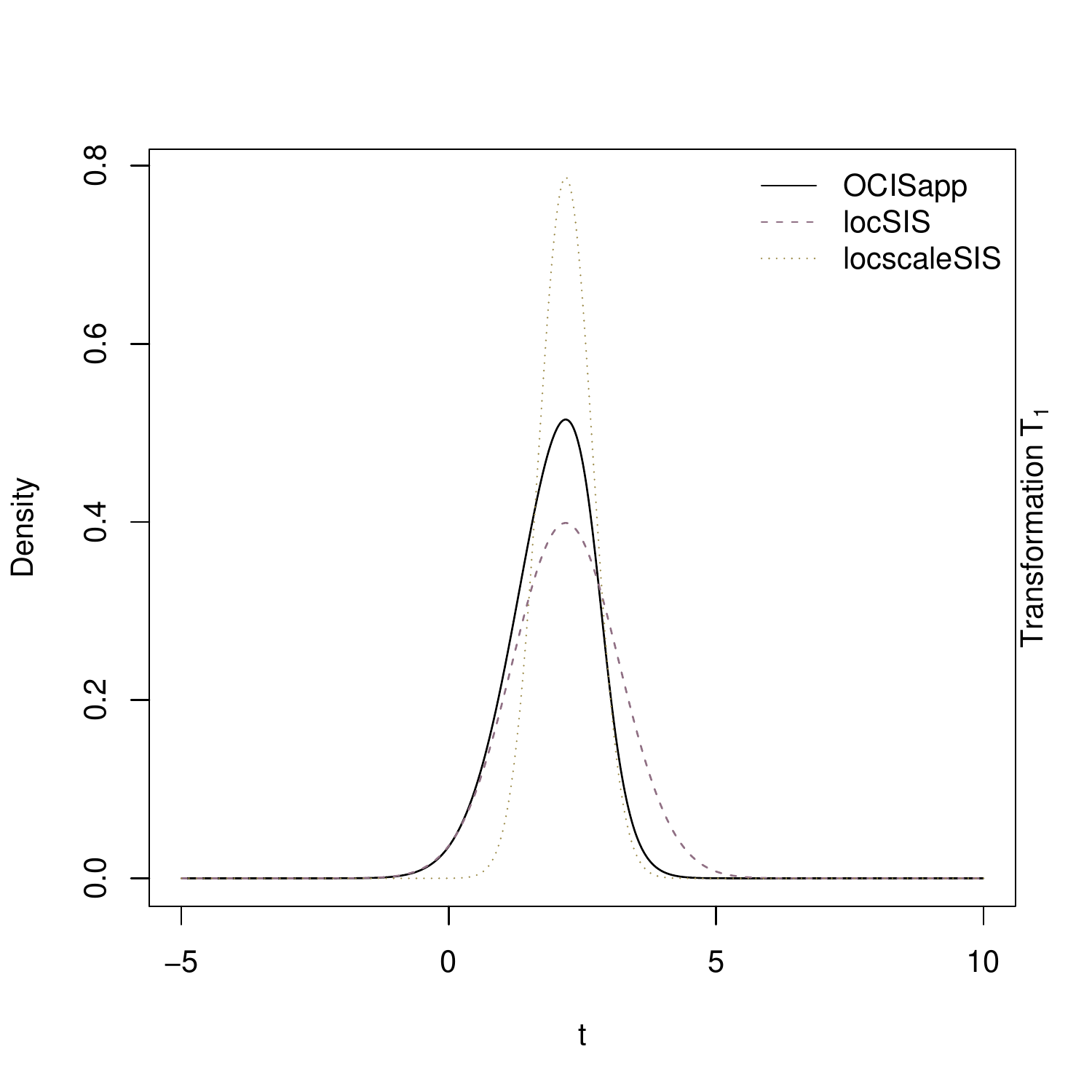}
\includegraphics[width = 0.48\textwidth]{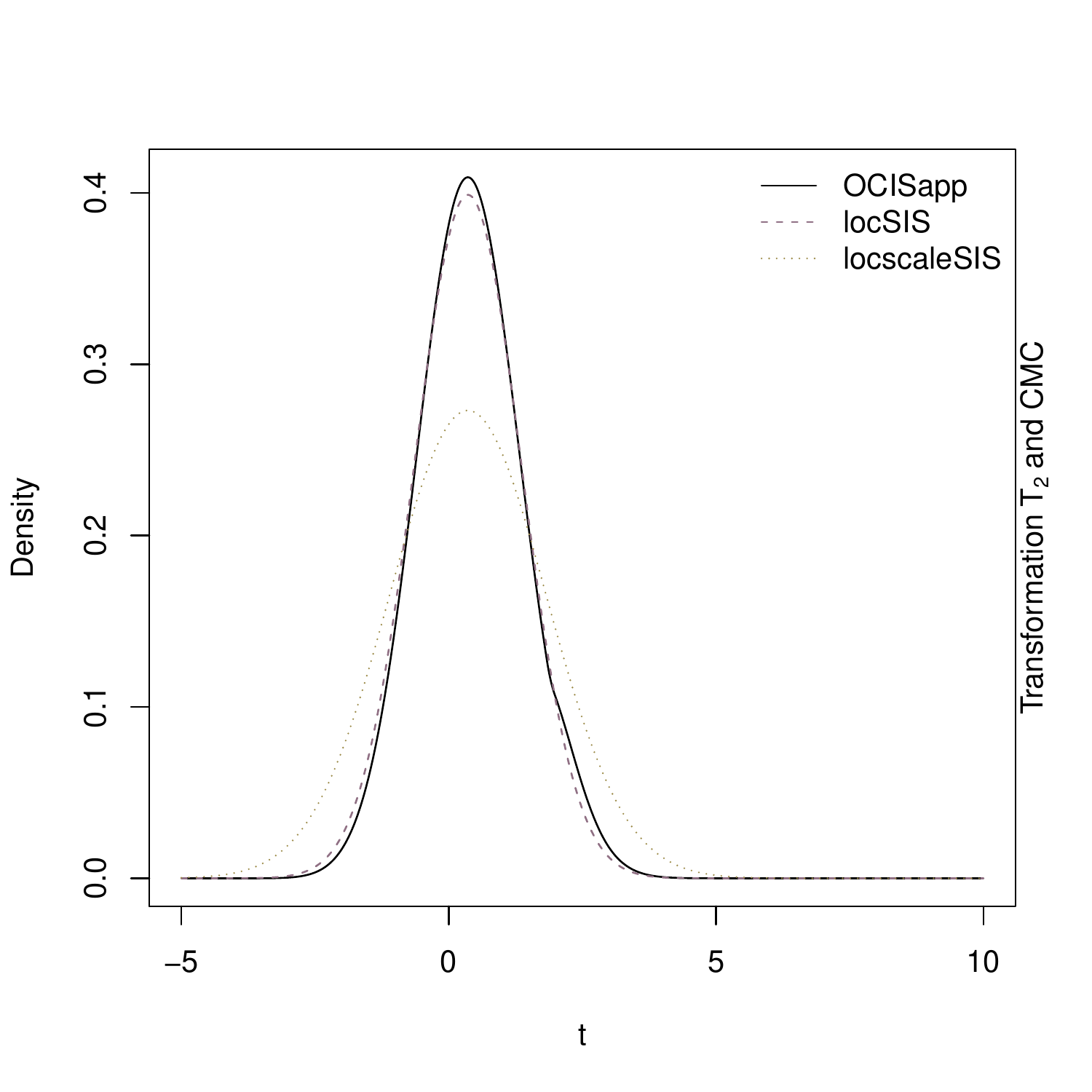}
\caption{Optimally calibrated densities for $l=100$ and the transformations $T_j$ for $j=1,2$. }
\end{figure}

\begin{figure}
\centering
\includegraphics[width = 0.45\textwidth]{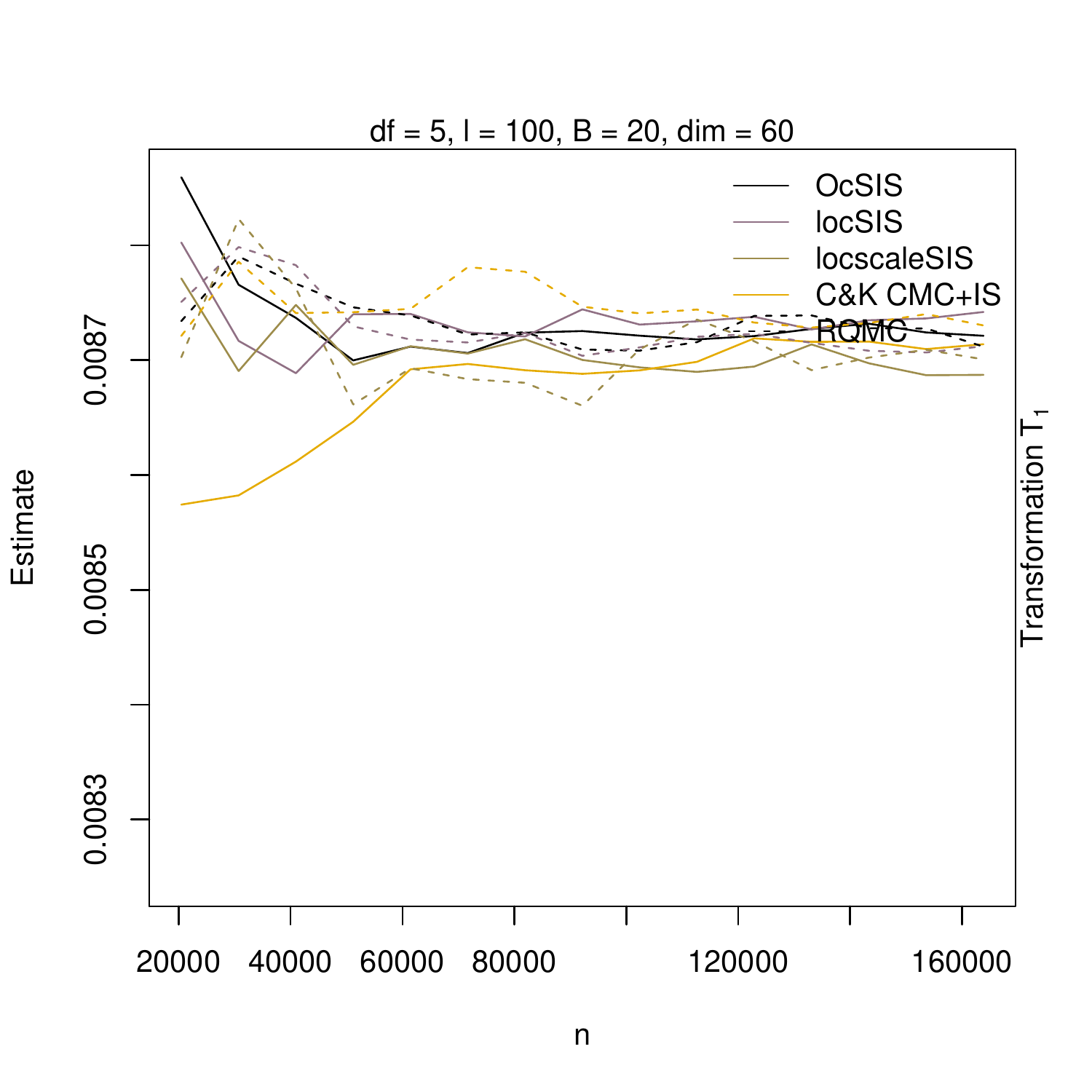}
\includegraphics[width = 0.45\textwidth]{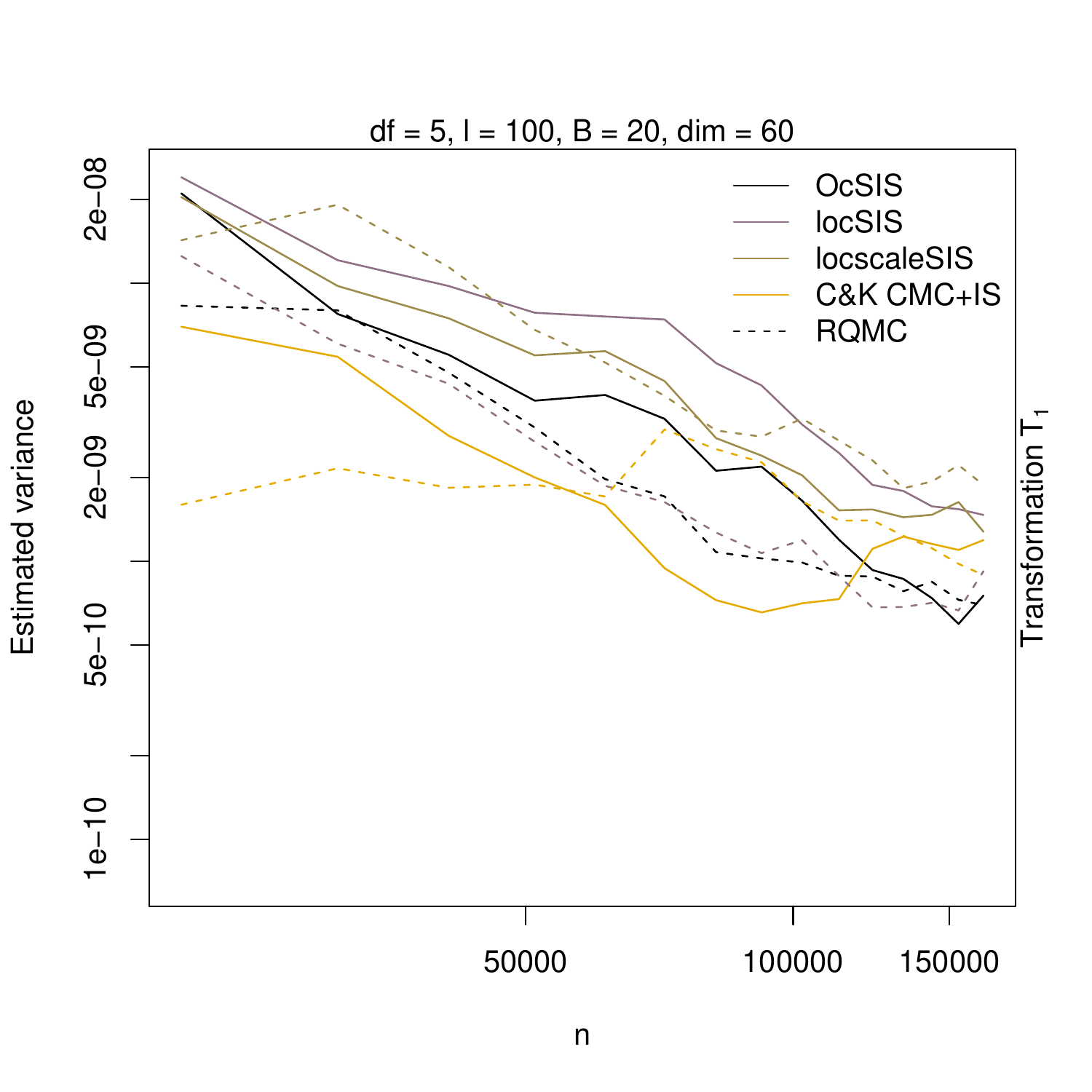}
\includegraphics[width = 0.45\textwidth]{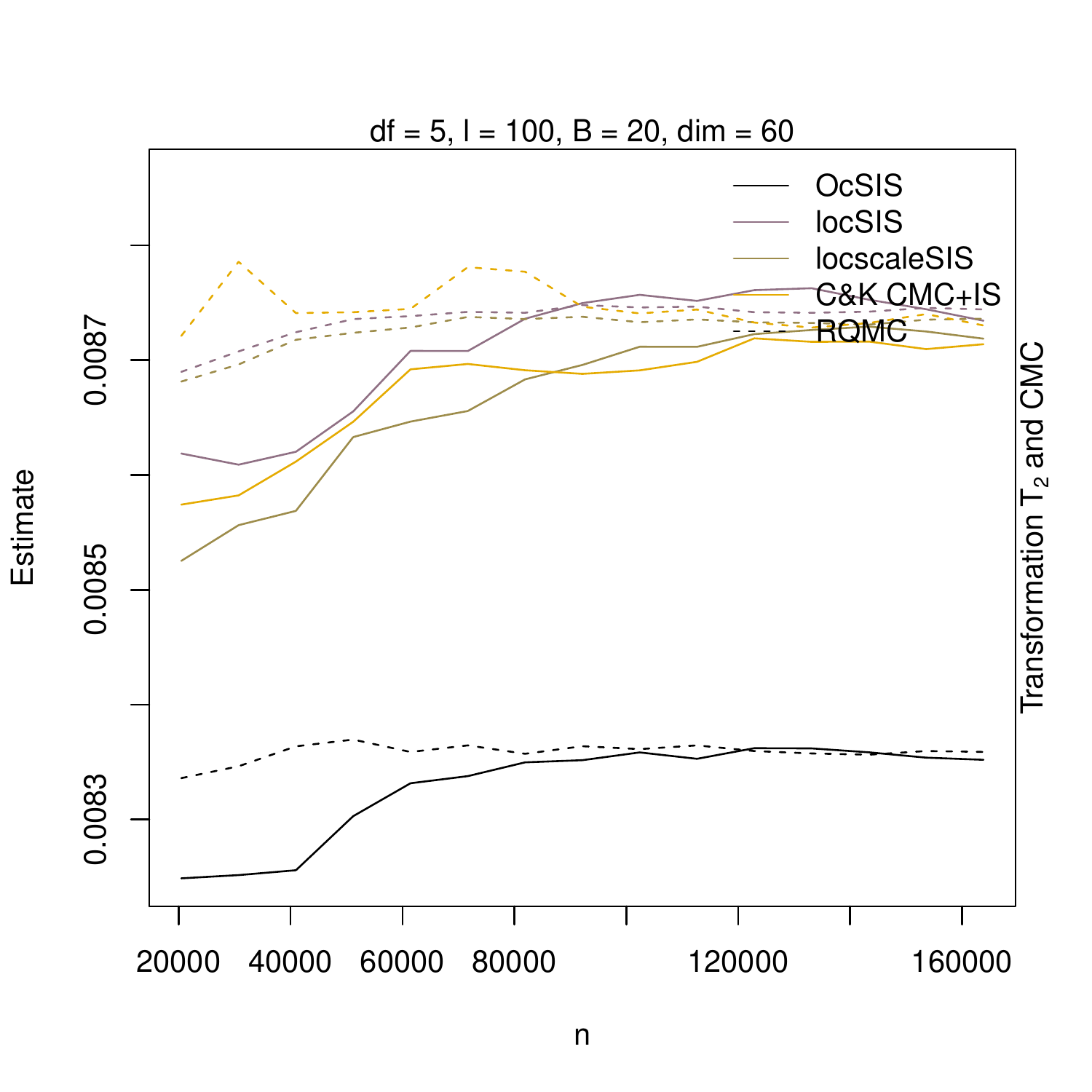}
\includegraphics[width = 0.45\textwidth]{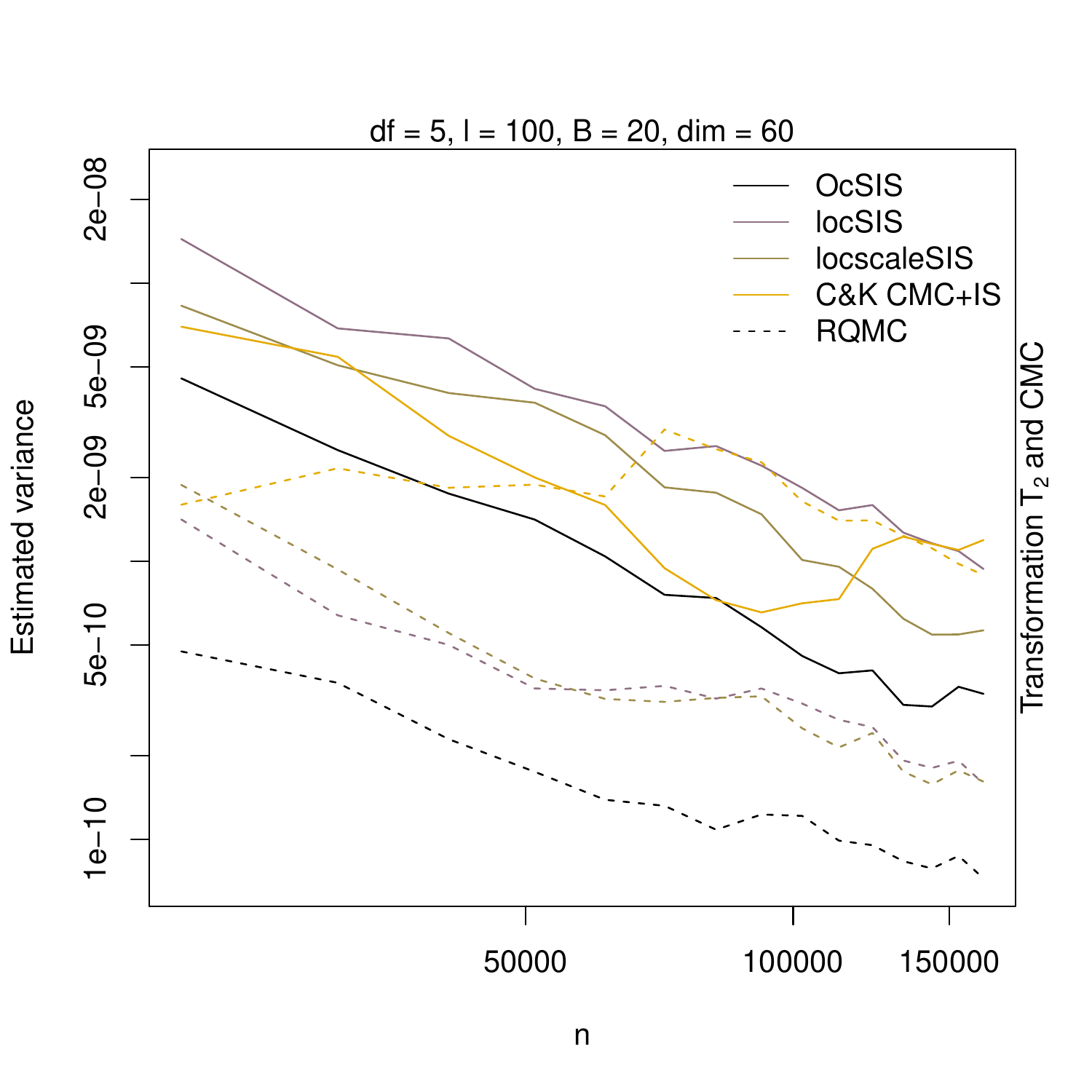}
\caption{Estimates (left) and estimated variances (right) as a function of $n$ for $\nu=5$. }
\label{fig:tcop:estvars:1}
\end{figure}

\begin{figure}
\centering
\includegraphics[width = 0.45\textwidth]{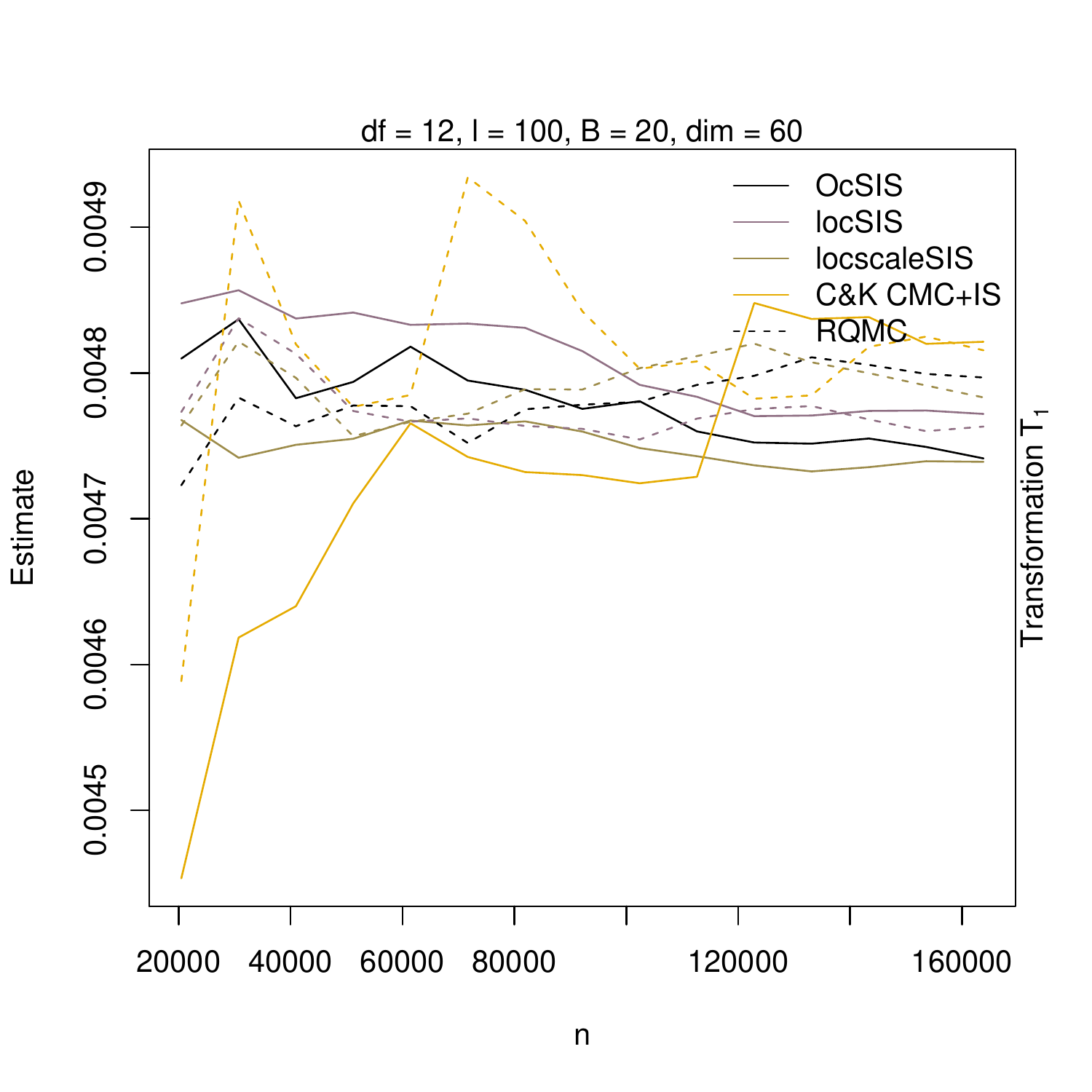}
\includegraphics[width = 0.45\textwidth]{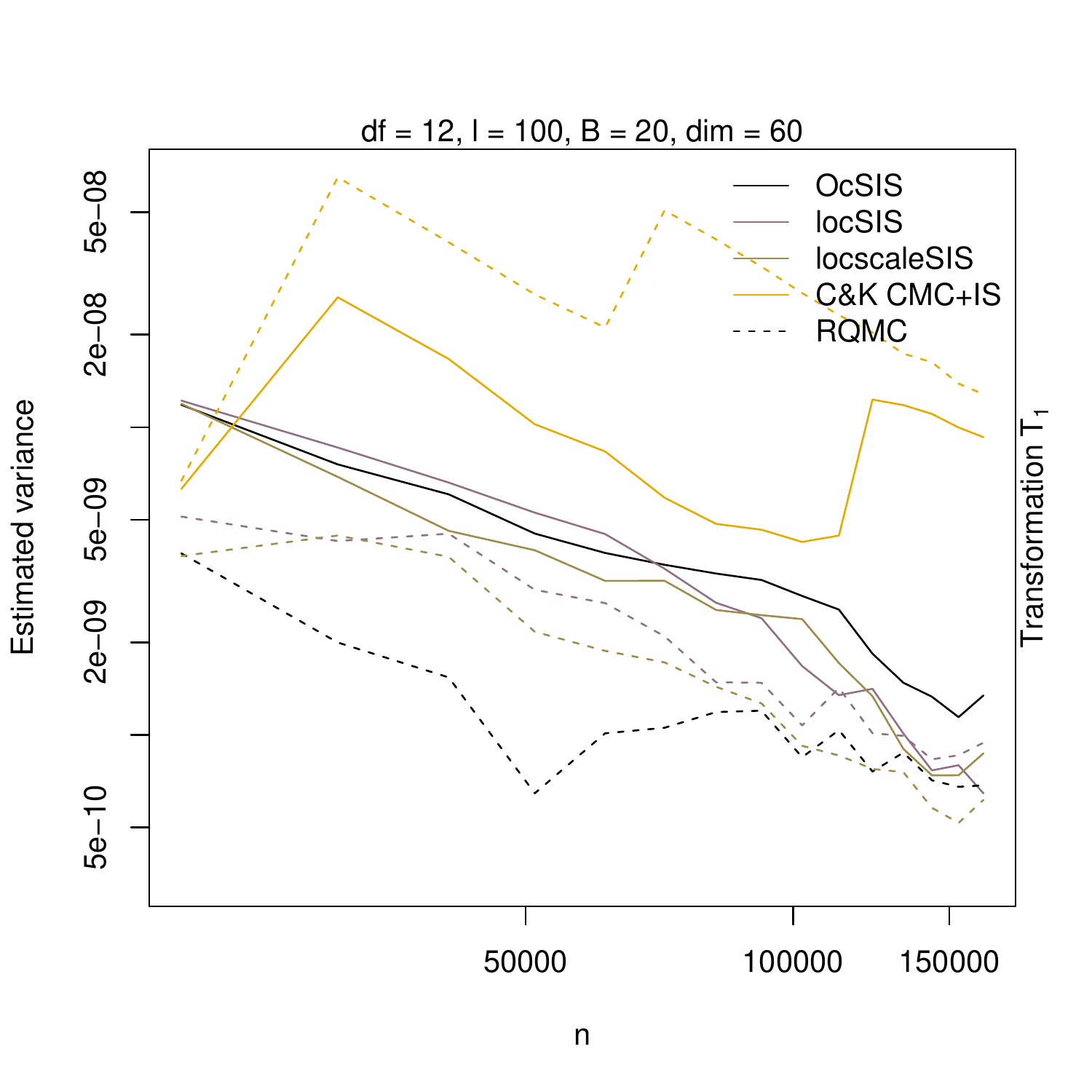}
\includegraphics[width = 0.45\textwidth]{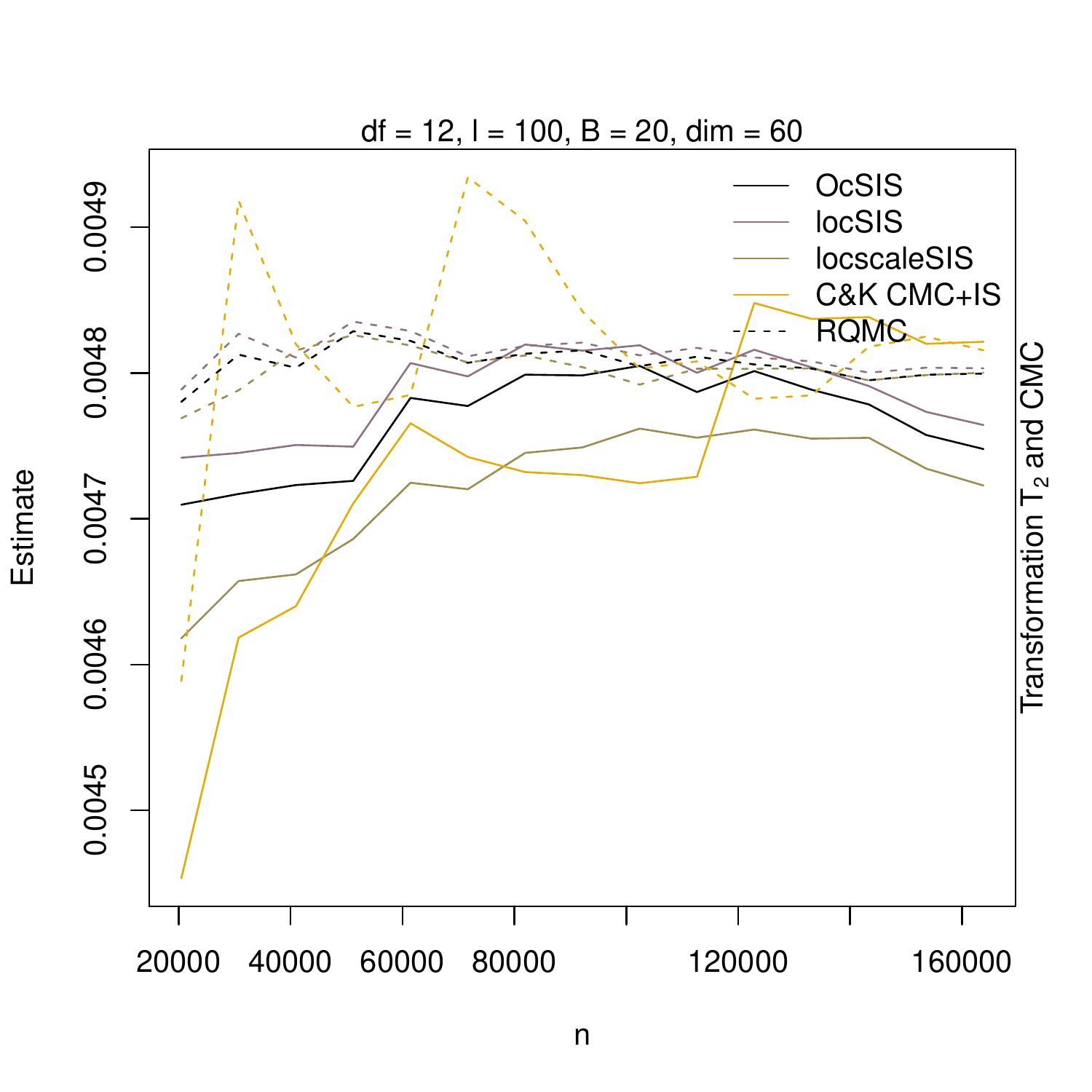}
\includegraphics[width = 0.45\textwidth]{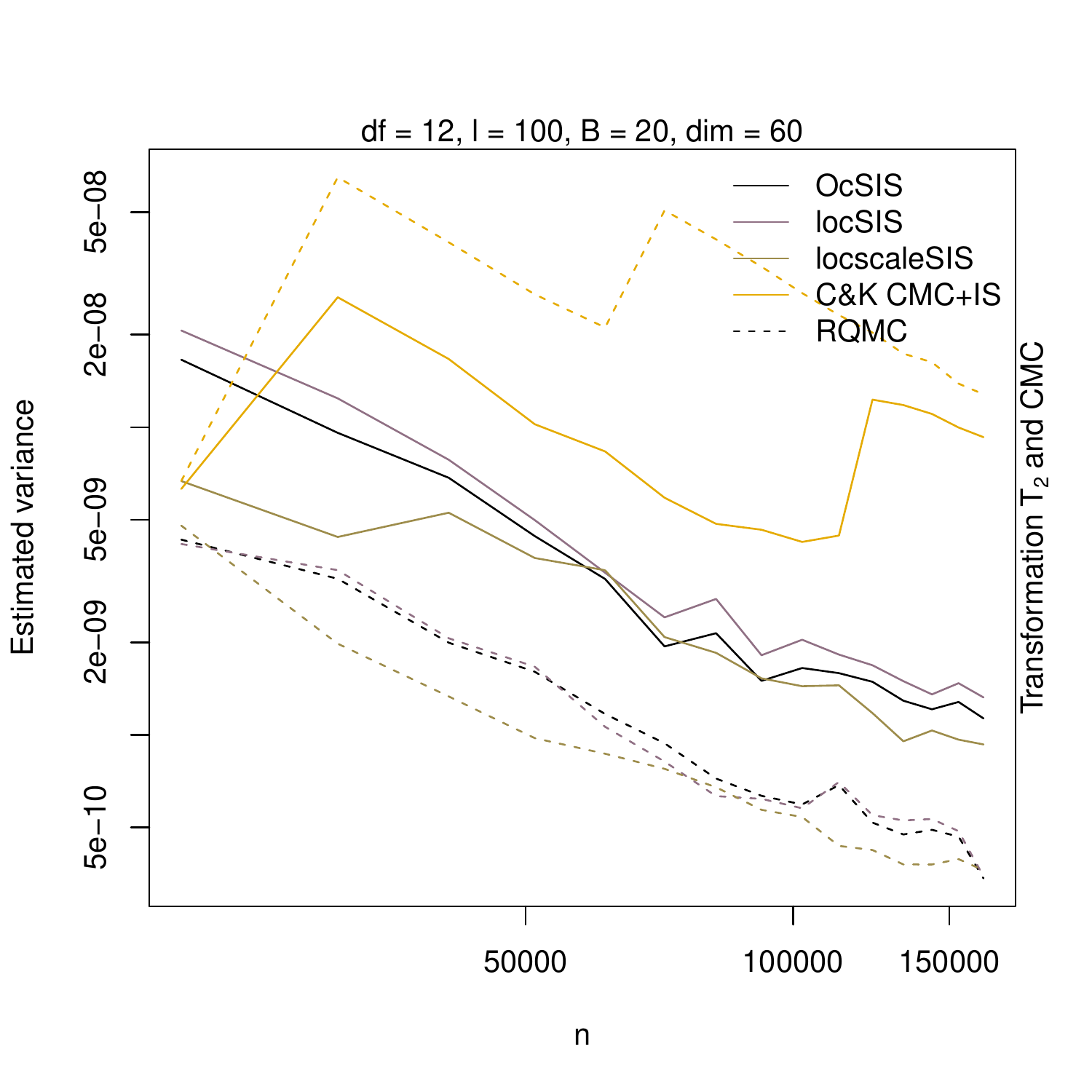}
\caption{Estimates (left) and estimated variances (right) as a function of $n$ for $\nu=12$. }
\label{fig:tcop:estvars:2}
\end{figure}

We compare the original C\&K CMC+IS from \cite{chankroese2010} with SIS with and without CMC. We additionally investigate whether employing RQMC yields a variance reduction. To this end, we estimate $p_l$ for $l=100$ for various $n$ and methods; see Figures~\ref{fig:tcop:estvars:1} and \ref{fig:tcop:estvars:2}. Variances are estimated as the sample variance of $B=20$ repetitions; this ensures that the same variance estimator (namely, the sample variance) is used for both methods, rather than using the estimator from Proposition~\ref{prop:SISvar:est} for MC and the sample variance for RQMC.

Note that for fixed $\nu$, the data for C\&K CMC+IS are identical independent of which transformation is used, so these lines can be used as reference. As expected, variances with the CMC idea are smaller than without the CMC idea. Note further that all our (S)SIS methods combined with $T_1$ (which does not integrate out $W$) give smaller variances than C\&K CMC+IS, which does integrate out $W$.

\section{Concluding Remarks}
\label{sec:conclusion}
In this paper, we developed importance sampling and stratification techniques that are designed to work well for problems with a single-index structure, i.e., where the response variable depends on input variables mostly through some one-dimensional transformation. The main theme of our approach is to exploit the low-dimensional structure of a given problem in rare-event simulation by introducing a conditional sampling step on this important transformed random variable and using optimal IS.

We derived expressions for optimal densities of said one-dimensional transformation which achieve minimum variance and discussed boundary cases with zero variance. Furthermore, we demonstrated that our framework includes and generalizes existing mean-shifting techniques.  Our theoretical framework and numerical examples suggest substantial variance reduction for problems having strong single-index structures. As the optimal density rarely belongs to a known parametric family, we also give explicit steps to calibrate the proposal distribution.

Our numerical experiments revealed that the proposed methods outperform existing methods that were specifically tailored to the Gaussian and $t$-copula credit portfolio problem. The success of our method in this framework highlights the flexibility and wide applicability of our approach.

By combining our single-index framework with RQMC methods, we achieve even more precise estimation results, thanks to the dimension reduction feature of our conditional sampling step.

Note that there exist many other low-dimensional structures studied in the literature and they may provide a better fit than single-index models do. For instance, the structure assumed by the sufficient dimension reduction can be seen as a multi-index extension of the linear single-index model; see \cite{cook1998, cookforzani2009, adrangicook2009}. We would like to develop importance sampling techniques for problems based on other low-dimensional structures in future research.

\appendix
\section*{Appendix}

\begin{proof}[Proof of Proposition~\ref{prop:IS}]
The mean and variance follow from
$$ \E_g(\hmuSISn) = E_g(\Psi(\bX) w(T)) = \E_g(m(T)w(T)) = \int_{\Omega_g} m(t) \frac{f_T(t)}{g_T(t)}g_T(t)\;\rd t=\mu_{\text{\tiny{SIS}}}$$
and
$$ n \var_g(\hmuSISn) + \mu_{\text{\tiny{SIS}}}^2 = \E_g\left(\Psi^2(\bX)w(T)\right)=\int_{\Omega_g} m^{(2)}(t)\frac{f_T^2(t)}{g_T^2(t)}\;\rd t.$$
Asymptotic normality follows from the central limit theorem. Next, we need to find $g_T$ among all $g$ that give unbiased estimators so that the variance, or equivalently $\E_g(m^{(2)}(T)w(T))$, is minimal when $\Psi(\bx)\geq 0$ or $\Psi(\bx)\leq 0$ for all $\bx\in\Omega$. Let $\Omega_{\text{\tiny{ub}}}=\{t\in\Omega_f:m(t)f_T(t)\ne 0\}$. By Jensen's inequality,
\begin{align*}
\E_g\left(m^{(2)}(T)w^2(T)\right)&\geq \left(\E_g\left(\sqrt{m^{(2)}(T)}w(T)\right)\right)^2 \\
&= \left(\int_{\Omega_g} \sqrt{m^{(2)}(t)}w(t)\;\rd t\right)^2 = \left(\int_{\Omega_f} \sqrt{m^{(2)}(t)}w(t)\;\rd t\right)^2
\end{align*}
The last inequality follows since $\hmuSISn$ is assumed to be unbiased, i.e., $\Omega_{\text{\tiny{ub}}}\subseteq \Omega_g$ and the fact that $\sqrt{m^{(2)}(t)}f_T(t)=0$ for $t\not\in\Omega_{\text{\tiny{ub}}}$ (as $m(t)=0$ implies $m^{(2)}(t)=0$ by the assumption on $\Psi$). The right hand side of the inequality is a constant independent of the choice of $g_T$, namely the minimum variance among all SIS estimators. To achieve equality, or equivalently to minimize the variance, set $g_T\propto \sqrt{m^{(2)}}(t)f_T(t)$ for $t\in\Omega_{\text{\tiny{ub}}}$ and the claim follows.
\end{proof}

\begin{proof}[Proof of Proposition~\ref{prop:SIS}]
Let $\Omega_T^{(i)}=\{t\in(t_{\inf},t_{\sup}): \lambda_{i}  \leq t < \lambda_{i+1}\}$ where $\lambda_i=G_T^\i( (i+1)/n)$ and note that $\P(T\in\Omega_T^{(i)})=1/n$ for $i=1,\dots,n$. Then
\begin{align*}
\E(\hmuSSISn)&=\frac{1}{n}\sum_{i=1}^n \E_g\left(\Psi(\bX)w(T)\mid T \in \Omega_T^{(i)}\right)=\frac{1}{n}\sum_{i=1}^n \E_g\left(\E_g(\Psi(\bX)w(T)\mid T) \mid T \in  \Omega_T^{(i)}\right)\\
&= \frac{1}{n}\sum_{i=1}^n  \int_{\lambda_i}^{\lambda_{i+1}}m(t) \frac{f_T(t)}{g_T(t)}g_T(t)\;\rd t = \mu_{\text{\tiny{SIS}}}
\end{align*}
The expression for the variance is a slight generalization of \cite[Lemma~4.1]{glassermanheidelbergershahabuddin1999} in that stratification is combined with IS, bit it can be proved similarly. Let $\eta_{n}(t)$ denote the index $i$ so that $t\in\Omega_T^{(i)}$. Then
\begin{align*}
n\Var(\hmuSSISn)=\frac{1}{n}\sum_{j=1}^n \Var_g\left(\Psi(\bX)w(T)\mid T \in\Omega_T^{(i)}\right)=\E_g\left(\Var_g\left(\Psi(\bX)w(T)\mid \eta_n(T)\right)\right).
\end{align*}
Let $\xi=\E_g(\Psi(\bX)w(T)\mid T)=m(T)w(T)$ and define the sequence $\xi_n=\E_g(\xi\mid \eta_n(T))$. Note that the $\sigma-$algebra generated by $\eta_n(T)$ forms an increasing family as $n$ increases through a constant multiple of power two. Observe that $\E_g(|\xi|)<\infty$ and $\sup_n \xi_n < \E_g(\Psi^(\bX)w^2(T))=\E_g(m^{(2)}(T)w^2(T))<\infty$. Also, $\xi_n$ is a martingale if $n$ increases through a constant multiple of powers of two as it is a Doob's martingale; see \cite[p.~246]{karlintaylor1975}. Then using the arguments as in \cite[Lemma~4.1]{glassermanheidelbergershahabuddin1999}, it follows that $\Var_g(\hmuSSISn)=\sigSISsq/n+o(1)$.

The expression for the optimal density and variance expressions follow as in the proof of Prop.~\ref{prop:IS} by applying Jensen's inequality. It remains to show that the SSIS estimator is asymptotically normal, which we show by applying the Lyapunov Central Theorem; see \cite[p.~134]{kolekoedijkverbeek2007}. Let $m_i = \E_g(\Psi(\bX)w(T)\mid T\in\Omega_T^{(i)})$ and $v_i^2=\Var_g(\Psi(\bX)w(T)\mid T\in \Omega_T^{(i)})$. It is easily seen that $(1/n)\sum_{i=1}^n m_i=\muSIS$ and $(1/n)\sum_{i=1}^n v_i^2=\sigSISsq+o(1)$. For any $i=1,\dots,n$, we have
\begin{align*}
&\E_g\left( |\Psi(\bX_i)w(T_i)-m_i|^{2+\delta}\right)\leq 2^{2+\delta}\left( \E_g\left(|\Psi(\bX_i)w(T_i)|^{2+\delta}\right)+\E_g\left(|m_i|^{2+\delta}\right)\right)\\
&= 2^{2+\delta} \left(\E_g\left( |\Psi(\bX)w(T)|^{2+\delta} \mid T \in \Omega_T^{(i)}\right)+ \E_g\left( | \E_g(\Psi(\bX)w(T)\mid T\in \Omega_T^{(i)}) |^{2+\delta}\right)\right)\\
&\leq 2^{2+\delta} \left(\E_g\left( |\Psi(\bX)w(T)|^{2+\delta} \mid T \in \Omega_T^{(i)}\right)+ \E_g\left( \E_g(|\Psi(\bX)w(T)|^{2+\delta}\mid T\in \Omega_T^{(i)}) \right)\right)\\
&=2^{3+\delta}\E_g\left( |\Psi(\bX)w(T)|^{2+\delta}\mid T\in\Omega_T^{(i)}\right),
\end{align*}
where the first inequality follows from the $c_{\tau}$ inequality as in \cite[p.~155]{loeve1963}. The Lyapunov condition is satisfied, since
\begin{align*}
&\frac{1}{(\sum_{i=1}^n\sigma_i^2)^{1+\delta/2}} \sum_{i=1}^n\E_g\left( |\Psi(\bX_i)w(T_i)-m_i|^{2+\delta}\right)\\
&\leq \frac{2^{3+\delta}}{(\sum_{i=1}^n\sigma_i^2)^{1+\delta/2}} \sum_{i=1}^n\E_g\left( |\Psi(\bX_i)w(T_i)|^{2+\delta}\mid T\in\Omega_T^{(i)}\right)\\
&= \frac{2^{3+\delta} n}{(n\sigSSISsq+o(n))^{1+\delta/2}} \E_g\left(|\Psi(\bX)w(T)|^{2+\delta}\right)\rightarrow0,\quad n\rightarrow\infty,
\end{align*}
by the assumption. The Lyapunov Central Limit Theorem together with Slutsky's Theorem implies $(\hmuSSISn-\muSIS)/\sqrt{n}\darrow\N(0,\sigSSISsq)$.
\end{proof}

\begin{proof}[Proof of Proposition~\ref{prop:SISvar:est}]
Recall that $T_i$ satisfies $T_i = G_T^\i( (i+U_i-1)/n)$ where $U_i\isim \U(0,1)$ for $i=1,\dots,n$, and are therefore ordered, i.e., $T_1<T_2<\dots<T_n$. For any $i=1,\dots,n$,
$$T_{i+1}-T_i=(G_T^{-1})'(\xi_i)\left(\frac{1+U_{i+1}-U_i}{n}\right)=\frac{1}{g_T(G_T^{-1}(\xi_i))}\left(\frac{1+U_{i+1}-U_i}{n}\right)=\mathcal{O}(1/n),$$
for some $\xi_i\in (T_i, T_{i+1})$, which implies that for any continuously differentiable function $h$, $h(T_{i+1})=h(T_i)+\mathcal{O}(1/n)$. Then we have
\begin{align*}
r_i^2&=\left( m(T_{i+1})+\eps_{T_{i+1}} - m(T_i) - \eps_{T_i}\right)^2\\
&= \left( m(T_{i+1})-m(T_i)\right)^2 + \left(\eps_{T_{i+1}}-\eps_{T_i}\right)^2-2(m(T_{i+1}) - m(T_i))(\eps_{T_{i+1}}-\eps_{T_i})\\
&= (\eps_{T_{i+1}}-\eps_{T_i})^2 - 2(m(T_{i+1}) - m(T_i))(\eps_{T_{i+1}}-\eps_{T_i})+\mathcal{O}(1/n^2),
\end{align*}
and so
\begin{align*}
\E_g\left(r_i^2 w^2(T_i)\right)&= \E_g\left(\E_g(r_i^2w(T_i)\mid T_i, T_{i+1})\right)=\E_g\left(w^2(T_i)(v^2(T_i)+v^2(T_{i+1}))\right)+\mathcal{O}(1/n^2)\\
&= 2 \E_g\left(w^2(T_i)v^2(T_i)\right)+\mathcal{O}(1/n),
\end{align*}
which means that
\begin{align*}
\E_g(\hsigSSISsq) &= \frac{1}{2(n-1)}\sum_{i=1}^n \E_g(r_i^2w^2(T_i))=\frac{1}{n}\sum_{i=1}^n\E_g\left(v^2(T)w^2(T)\mid T\in\Omega_T^{(i)}\right)+\mathcal{O}(1/n)\\
&= \E_g(v^2(T)w^2(T)) + \mathcal{O}(1/n) = \sigSSISsq + \mathcal{O}(1/n)\rightarrow \sigSSISsq,
\end{align*}
which shows consistency.
\end{proof}

\begin{proof}[Proof of Proposition~\ref{prop:meanshift}]
We use that $(\bX \mid T = t) \sim \N_d(\bbeta t,I_d-\bbeta\bbeta\T)$(see \cite[Theorem 1]{harrishelvig1965}) to compute the moment generating function of $\bX$. For $\ba\in\IR^d$,
\begin{align*}
\phantom{\E_g(}&\E_g(\exp(\ba\T\bX)) = \E_g \left(\E(\exp(\ba\T\bX) \mid T)\right) =
\E_g\left(\exp\left(\ba\T\bbeta T + \frac{1}{2} \ba\T(I_d-\bbeta\bbeta\T)\ba\right)\right) \\
&=\E_g(\exp(\ba\T\bbeta T ))\exp\left(\frac{1}{2} \ba\T(I_d-\bbeta\bbeta\T)\ba)\right)=\exp\left(c\ba\T\bbeta+\frac{1}{2}(\ba\T\bbeta)^2\sigma^2\right) \times\\
&\times \exp\left(\frac{1}{2} \ba\T(I_d-\bbeta\bbeta\T)\ba)\right) = \exp\left(\ba\T(c\bbeta) + \frac{1}{2}\ba\T (I_d+(\sigma^2-1)\bbeta\bbeta\T\ba\right).
\end{align*}
By uniqueness of the moment generating function, $\bX\sim \N_d(c\bbeta, I_d+(\sigma^2-1)\bbeta\bbeta\T)$.
\end{proof}

\printbibliography[heading=bibintoc]
\end{document}
